\documentclass{article}

\usepackage{arxiv}
\usepackage[utf8]{inputenc}
\usepackage[T1]{fontenc}
\usepackage{booktabs}
\usepackage{amsfonts}
\usepackage{nicefrac}
\usepackage{microtype}
\usepackage{lipsum}
\usepackage{graphicx}           
\usepackage{xcolor}
\usepackage{amsthm}
\usepackage{mathtools}
\usepackage{wrapfig}
\usepackage{caption}
\usepackage{natbib}
\usepackage{ulem}
\usepackage{rotating}           
\usepackage{adjustbox}
\usepackage{comment}
\usepackage{algorithm}
\usepackage{algpseudocode}
\usepackage{xurl}               
\usepackage[                    
  colorlinks = true,
  urlcolor   = blue,
  linkcolor  = black,
  citecolor  = black
]{hyperref}

\graphicspath{{./images/}}

\newtheorem{theorem}{Theorem}

\newtheorem{rem}{Remark}

\newcommand{\de}{\,\mathrm{d}}

 \title{\bf Change-point analysis: a new perspective for unstable financial markets}
  \author{\href{https://orcid.org/0000-0002-1964-7539}{\includegraphics[scale=0.1]{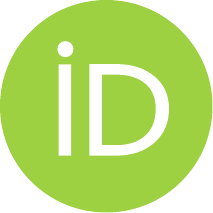}}\hspace{1mm} \v Zikica Luki\' c
    \hspace{.2cm}\\
    Independent Researcher\\
    Present address: Biometrics FSP, Parexel International
    \\
    \texttt{zikicamaster@gmail.com}
    \And
    \href{https://orcid.org/0000-0001-8243-9794}{\includegraphics[scale=0.1]{orcid.pdf}}\hspace{1mm}Bojana Milo\v sevi\' c\thanks{
    The work of B. Milo\v sevi\' c is supported by the Ministry of Science, Technological Development and Innovations of the Republic of Serbia (the contract 451-03-136/2025-03/200104).The results are also obtained upon work from COST Action HiTEc-Text, functional and other high-dimensional data in econometrics: New models, methods, applications, CA21163, supported by COST (European Cooperation in Science and Technology)}\hspace{.2cm}  \\
    Faculty of Mathematics, University of Belgrade\\
    \texttt{bojana@matf.bg.ac.rs}}

\begin{document}
\maketitle
\begin{abstract}
We introduce two new classes of nonparametric change-point tests for sequences of univariate non-negative random variables. The proposed procedures are based on the empirical modified Hankel transform and the Laplace transform, respectively, and provide new transform-based tools for detecting distributional changes. We derive the asymptotic null distributions of the corresponding test statistics and use a permutation bootstrap procedure to obtain $p$-values, since the limiting null distributions are not distribution-free. Through a finite-sample simulation study, we show that the proposed tests are well calibrated, generally more powerful than empirical-characteristic-function-based competitors, and capable of accurately estimating the change-point location, particularly in asymmetric settings. The usefulness of the proposed methodology is further demonstrated through applications to Argentina rainfall data, as well as U.S. GNP and S\&P 500 absolute log-returns. The real-data examples illustrate that the proposed Hankel- and Laplace-transform-based tests are effective tools for detecting meaningful distributional changes in non-negative data.
\end{abstract}

Change-point analysis has attracted significant attention in recent years due to its wide range of applications, particularly in fields such as medicine \cite{ hall2000change, minard2022change, yu2026change}, oceanography \cite{killick2012optimal}, climatology \cite{lund2023good}, and finance \cite{kim2022unsupervised, habibi2021bayesian, thies2018bayesian}. {In economics, change-point methods can be used to reduce the incidence of bad bets in momentum-trading strategies \cite{wood2022slowmomentum}. Moreover, change-point analysis has been applied to U.S. GNP data in \cite{wu2025trend}, while numerous studies have examined the influence of structural breaks on monetary policy (see, e.g., \cite{hauzenberger2022fast, koop2012forecasting}). }

Change-point analysis in statistics is typically approached through two main frameworks: parametric and nonparametric methodologies. Parametric methods assume a specific model and detect change-points based on changes in the model parameters \cite{bai1997estimation, goodman2006survival, gupta2001change}. By contrast, nonparametric methods do not rely on underlying assumptions about the model. These methods mainly focus on detecting change-points in the distribution of the observations \cite{austin2023online, carlstein1988nonparametric, matteson2014nonparametric}.

A commonly used strategy in nonparametric change-point inference involves computing the maximum of the integrated difference between suitably chosen integral transforms applied to the first $k$ and the remaining $n-k$ observations. This maximum is used both as a test statistic and as an estimator of the change-point location. The methodology was first introduced using the characteristic function in \cite{huvskova2006change} and was later adapted to the matrix-variate setting based on orthogonal Hankel transforms in \cite{nas2}. Recent methodological advances in univariate change-point detection include \cite{boniece2025sequential, hlavka2026change, ratnasingam2026single}, among others.

In this paper, we extend this general framework to sequences of univariate non-negative random variables by developing two novel classes of change-point test statistics based on the empirical modified Hankel transform introduced in \cite{baringhaus2015two, baringhaus2010empirical} and the Laplace transform. Although both transforms are well known and characterize the distributions of non-negative random variables, they have not previously been explored in this context, and they lead to new families of test statistics with favorable theoretical and empirical properties. In particular, the matrix-based statistic proposed in \cite{nas2} is based on orthogonal Hankel transforms and cannot be viewed as a generalization of the univariate Hankel-transform-based statistic developed here. Thus, our contribution is both conceptually and technically distinct and enriches the set of available nonparametric tools for change-point detection.

In Section \ref{sec::teststat}, the test statistics are presented together with their asymptotic properties. Section \ref{sec::power} presents the power study, while Section \ref{sec::realdata} provides real-data examples that illustrate applications of the proposed methodology. { All tables are provided in Appendix \ref{appendixA}.}

\section{Test statistics}\label{sec::teststat}

Let $ X_1, X_2, \dots, X_n $ be independent non-negative random variables, where each  $X_j$ has a distribution function $F_j$, for $ j \in \{1, 2, \dots, n\} $.
 We consider testing $H_0: F_1=F_2=\dots = F_n$ against the alternative $H_1: F_1=F_2=\dots = F_{k}\neq F_{k+1}=F_{{k+2}}=\dots=F_n$, where $k, F_1$, and $F_n$ are unknown. We consider the following two classes of test statistics, motivated by the two-sample test proposed  in \cite{baringhaus2015two} and the general approach in \cite{huvskova2006change}:
\begin{align}
    \label{hankelstat}\mathcal{J}_{n, \gamma,a}& = \max_{1\leq k < n}\Big[ \Big(\frac{k(n-k)}{n^2}\Big)^\gamma \frac{k(n-k)}{n}\int\limits_0^\infty (\mathcal{H}_{X}(t)-\mathcal{H}^0_{X}(t))^2\exp(-at)\de t\Big], \; a > 0; \\
      \label{laplacestat}\mathcal{L}_{n, \gamma, a}& = \max_{1\leq k < n} \Big[\Big(\frac{k(n-k)}{n^2}\Big)^\gamma \frac{k(n-k)}{n}\int\limits_0^\infty (\mathcal{L}_{X}(t)-\mathcal{L}^0_{X}(t))^2\exp(-at)\de t\Big], \;a>0,
\end{align}
where 
\begin{align*}
 \mathcal{H}_{X}(t)=\frac{1}{k}\sum_{j=1}^k J_0(2\sqrt{tX_j}),\quad  & \mathcal{L}_{X}(t)=\frac{1}{k}\sum_{j=1}^k\exp(-tX_{j}) \\
 \mathcal{H}^{0}_{X}(t)=\frac{1}{n-k}\sum^n_{j=k+1} J_0(2\sqrt{tX_j}),\quad & \mathcal{L}^{0}_{X}(t)=\frac{1}{n-k}\sum^n_{j=k+1}\exp(-tX_{j})\\
\end{align*}
are the empirical Hankel and Laplace  transforms of $X_1, X_2, \dots, X_k$, and  of $X_{k+1}, X_{k+2}, \dots, X_n$, respectively. Here, $J_0$ is the Bessel function of the first kind of order zero.

\begin{rem}
The approach proposed in \eqref{hankelstat} and \eqref{laplacestat} can be extended by considering alternative weight functions beyond \( w(t) = \exp(-at) \). 
The choice of weight function influences the resulting test statistics and may affect their finite-sample performance. In this work, we focus on the exponential weight function because it yields closed-form expressions for the integrals involved, leading to greater analytical tractability and computational efficiency.
\end{rem}

Using results from \cite{baringhaus2015two} and \cite[Eq. 6.615]{gradshteyn1988tables}, we obtain the following closed-form expression for  $ \mathcal{J}_{n, \gamma,a} $:
\begin{align*}
          \mathcal{J}_{n, \gamma,a} = &\max_{1\leq k < n} \Big[\Big(\frac{k(n-k)}{n^2}\Big)^\gamma \frac{k(n-k)}{na}\Big(\frac{1}{k^2}\sum\limits_{l=1}^k\sum\limits_{m=1}^k I_0\left(\frac{2\sqrt{X_lX_m}}{a}\right)\exp\left(\frac{-(X_l+X_m)}{a}\right)+\\&\frac{1}{(n-k)^2}\sum\limits_{l=k+1}^n\sum\limits_{m=k+1}^n I_0\left(\frac{2\sqrt{X_lX_m}}{a}\right)\exp\left(\frac{-(X_l+X_m)}{a} \right)-\\&\frac{2}{k(n-k)}\sum\limits_{l=1}^k\sum\limits_{m=k+1}^n I_0\left(\frac{2\sqrt{X_lX_m}}{a}\right)\exp\left(\frac{-(X_l+X_m)}{a}\right) \Big)\Big],
\end{align*}
where $I_0$ is the modified Bessel function of the first kind of order $0$. Direct computations yield the following form for $\mathcal{L}_{n, \gamma, a}$:
\begin{align*}
          \mathcal{L}_{n, \gamma, a} \!=\! &\max_{1\leq k < n} \Big[\Big(\frac{k(n-k)}{n^2}\Big)^\gamma \frac{k(n-k)}{n}\Big(\frac{1}{k^2}\sum\limits_{l=1}^k\sum\limits_{m=1}^k \frac{1}{X_l+X_m+a}\!+\!\frac{1}{(n-k)^2}\sum\limits_{l=k+1}^n\sum\limits_{m=k+1}^n \frac{1}{X_l+X_m+a}\\&\!-\frac{2}{k(n-k)}\sum\limits_{l=1}^k\sum\limits_{m=k+1}^n \frac{1}{X_l+X_m+a} \Big)\Big].
\end{align*}

The following theorem establishes the asymptotic behavior of the test statistics under the null hypothesis. The proof follows the methodology outlined in \cite{huvskova2006change} and \cite{nas2}. For the general case of an arbitrary integral transform, the proof is provided in the Appendix \ref{appendixB}.



Let $q_{\mathcal{J}}(x, y;a) = {\frac{1}{a}}I_0\big(\frac{2\sqrt{xy}}{a}\big)\exp(-\frac{x+y}{a})$ and $q_{\mathcal{L}}(x, y;a) = \frac{1}{x+y+a}$. Let us define $\Tilde{q}_{I} (x, y;a)= q_{I}(x, y;a)- Eq_{I}(x, X_s;a)-Eq_{I}(X_r, y;a)+Eq_{I}(X_r, X_s;a),\;I\in \{\mathcal{J},\mathcal{L}\},$ for $r\neq s$. {For $I\in\{\mathcal J,\mathcal L\}$, let $(\lambda_j^I)_{j\geq1}$ denote the nonnegative eigenvalues, repeated according to multiplicity and arranged in non-increasing order, of the integral operator on $L^2(F)$ with kernel $\widetilde q_I(\cdot,\cdot;a)$, where integration is taken with respect to $F$.} The following theorem holds.

\begin{theorem}\label{asimptotikaM}
    Let $X_1, X_2, \dots, X_n$ be independent {and equally distributed (IID)} non-negative random variables, where $X_1$ has a distribution function $F$. Let $\gamma \in (0, 1]$ and let $(\lambda_j^{I})_{j=1}^\infty$ be the {non-increasing} sequence of eigenvalues defined above.
Then the asymptotic distributions of $\mathcal{J}_{n, \gamma,a}$ and $\mathcal{L}_{n, \gamma,a}$ are equal to
    \begin{equation}\label{asimptotika1}
        \sup\limits_{t\in (0, 1)}\Big[(t(1-t))^\gamma \Big|\frac{1}{a}EI_0\big(\frac{2X_1}{a}\big)\exp(-\frac{2X_1}{a})-\frac{1}{a}E I_0\big(\frac{2\sqrt{X_1X_2}}{a}\big)\exp\big(-\frac{X_1+X_2}{a}\big)+\sum\limits_{j=1}^\infty \lambda^{\mathcal{J}}_j\Big(\frac{B_j^2(t)}{t(1-t)}-1\Big)\Big|\Big], 
    \end{equation}
and
\begin{align}\label{asimptotika2}
    \sup\limits_{t\in (0, 1)}\Big[(t(1-t))^\gamma \Big|E\Big(\frac{1}{a+2X_1}\Big)-E\Big(\frac{1}{a+X_1+X_2}\Big)+\sum\limits_{j=1}^\infty \lambda^{\mathcal{L}}_j\Big(\frac{B_j^2(t)}{t(1-t)}-1\Big)\Big|\Big], 
    \end{align}

    
   %
  respectively, where $\{B_{j}(t),\; t\in (0,1)\}, j=1, 2, \dots$ are independent Brownian bridges.
\end{theorem}

Since the limiting null distributions of the proposed statistics are not distribution-free, we suggest using the permutation bootstrap algorithm from \cite{huvskova2006change} to derive p-values. This approach can be theoretically justified for our tests in a similar manner.

The theoretical development in this paper is carried out under the assumption of independent observations. This framework is natural for the proposed methodology, since our aim is to construct fully nonparametric tests for detecting distributional changes in sequences of non-negative random variables without imposing a parametric model for the data-generating mechanism. Under the null hypothesis, the observations are identically distributed and independent, and hence exchangeable. This exchangeability provides the basis for using a permutation bootstrap to approximate the null distribution of the proposed statistics.
  
\section{Power study}\label{sec::power}

In this section, we investigate the finite-sample properties of the proposed tests.

As benchmarks, we consider the empirical-characteristic-function-based statistics proposed in \cite{huvskova2006change}:

\begin{align*}
      T^{(1)}_{n, \gamma,a} =2a &\max_{1\leq k < n} \Big[\Big(\frac{k(n-k)}{n^2}\Big)^\gamma \frac{k(n-k)}{n}\Big(\frac{1}{k^2}\sum\limits_{l=1}^k\sum\limits_{m=1}^k \frac{1}{a^2+(X_l-X_m)^2}+\\&\frac{1}{(n-k)^2}\sum\limits_{l=k+1}^n\sum\limits_{m=k+1}^n \frac{1}{a^2+(X_l-X_m)^2}-\\&\frac{2}{k(n-k)}\sum\limits_{l=1}^k\sum\limits_{m=k+1}^n \frac{1}{a^2+(X_l-X_m)^2} \Big)\Big].
\end{align*}
\begin{align*}
      T^{(2)}_{n, \gamma,a} =\sqrt{\frac{\pi}{2}} &\max_{1\leq k < n}\Big[\Big(\frac{k(n-k)}{n^2}\Big)^\gamma \frac{k(n-k)}{n}\Big(\frac{1}{k^2}\sum\limits_{l=1}^k\sum\limits_{m=1}^k \exp\Big(-\frac{(X_l-X_m)^2)}{4a}\Big)+\\&\frac{1}{(n-k)^2}\sum\limits_{l=k+1}^n\sum\limits_{m=k+1}^n \exp\Big(-\frac{(X_l-X_m)^2)}{4a}\Big)-\\&\frac{2}{k(n-k)}\sum\limits_{l=1}^k\sum\limits_{m=k+1}^n \exp\Big(-\frac{(X_l-X_m)^2)}{4a}\Big) \Big)\Big].
\end{align*}

Note that, under $H_0$, the test statistics considered here, including both the proposed statistics and the competitors, are not distribution-free; therefore, we use the permutation bootstrap algorithm with $B=500$ replicates to estimate the p-values for these tests. We use $N=2000$ Monte Carlo replicates to estimate the empirical powers of the tests. 



Following the scenarios outlined in \cite{huvskova2006change}, we consider sample sizes $n=40$ and $n=100$, with the significance level set to $\alpha=0.05$. The subsamples $X_1, X_2, \dots, X_k$ and $X_{k+1}, X_{k+2}, \dots, X_n$ are generated from distributions $F_1$ and $F_n$, respectively, where $k=\frac{n}{2}$ or $k=\frac{n}{4}$, $F_1$ is either the uniform distribution $U[0,1]$ (U), the gamma distribution $\Gamma(1, 1)$ (G1), or the gamma distribution $\Gamma(2, 1)$ (G2), and $F_n$ is the distribution of $bX_1+1$.

The powers for $U[0, 1]$ are equal to 1 and are therefore omitted from the results; the remaining powers, expressed as percentages, are presented in Tables \ref{PWOR}--\ref{PWOR100G2}. In the tables, * denotes a test power of 100 percent.

The results show that all tests are well calibrated. It is clear that the proposed tests $\mathcal{J}$ and $\mathcal{L}$ exhibit higher power than the benchmark tests. The parameter $\gamma$ appears to have a greater influence on the powers of the tests when the change-point is located in the middle of the sample. Moreover, the actual location of the change-point significantly influences the empirical powers, which peak in the symmetric case. For both proposed tests, the tuning parameter $a$ in the weight function has an impact, and no universal conclusion can be drawn. { However, for both tests, $a = 1$ appears to offer a reasonable compromise, providing stable performance across different scenarios.}

To assess the quality of the change-point estimate, we generated $N=10,000$ sample replications and estimated the change-point using all statistics included in the power study. We present the mean and the standard deviation of the estimated position in Tables \ref{Gamma05Pozicija40} and \ref{Gamma05Pozicija100}. The results indicate that the proposed tests yield more precise estimates when the true position of the change-point is at $\frac{1}{4}$ of the sample, whereas for $\frac{1}{2}$ of the sample, the proposed tests are comparable to the benchmarks, with smaller standard deviations for $G_1$ and $G_2$. The estimation quality improves with increasing sample size, as expected.

\section{Real data examples}\label{sec::realdata} 

In this section, we apply the proposed tests to several types of data. The examples cover macroeconomic, environmental, and financial time series, illustrating the practical usefulness of the proposed methodology in detecting meaningful distributional changes. The proposed tests are designed for independent non-negative observations. Accordingly, the real-data examples below are analyzed within this framework. In each application, the observations are indexed by their natural chronological order, and the objective is to detect departures from distributional homogeneity along this index. Thus, the resulting change-points should be interpreted as evidence of changes in the marginal distribution of the transformed observations, rather than as estimates from a fully specified dependent time-series model.

\subsection{U.S. Gross National Product (GNP) Data}

In this subsection, we analyze a data set consisting of quarterly U.S. GNP, in billions of chained 2017 dollars, from 1947(1) to 2023(3); the data have been seasonally adjusted \cite{gnpdata}. The first difference of the logarithm of GNP can be naturally interpreted as the GNP growth rate \cite{shao2010testing}. In \cite{shumway2006time}, the authors consider a shorter time frame, from 1947 (1) to 2002 (3), in chained 1996 dollars, also seasonally adjusted. The authors examine the difference in the logarithm of GNP and conclude that there may be a structural break in the data in 1985. The authors in \cite{shao2010testing} reach the same conclusion, using the same data set as in \cite{shumway2006time}. 

However, our tests are designed for non-negative data. Therefore, we consider the absolute difference in the logarithm of GNP. The data are shown in Figure \ref{fig:GNP}. We performed single-change-point detection using the permutation bootstrap with $B=1,000$ replications to obtain $p$-values. The results are presented in Table \ref{pvalUSGNP}. The proposed tests report $p$-values smaller than 0.05. Our results for both proposed tests are consistent with those in \cite{shao2010testing,shumway2006time} and provide further evidence that a structural break occurred in 1985. Moreover, the competing tests also detect the existence of a change-point, although the estimated location is slightly earlier, which may be explained by their lower precision, as demonstrated in Tables \ref{Gamma05Pozicija40} and \ref{Gamma05Pozicija100}.

\begin{figure}[H]
\caption{Quarterly U.S. GNP absolute growth rate from 1947(1) to 2023(3)}
\includegraphics[width=0.6\textwidth]{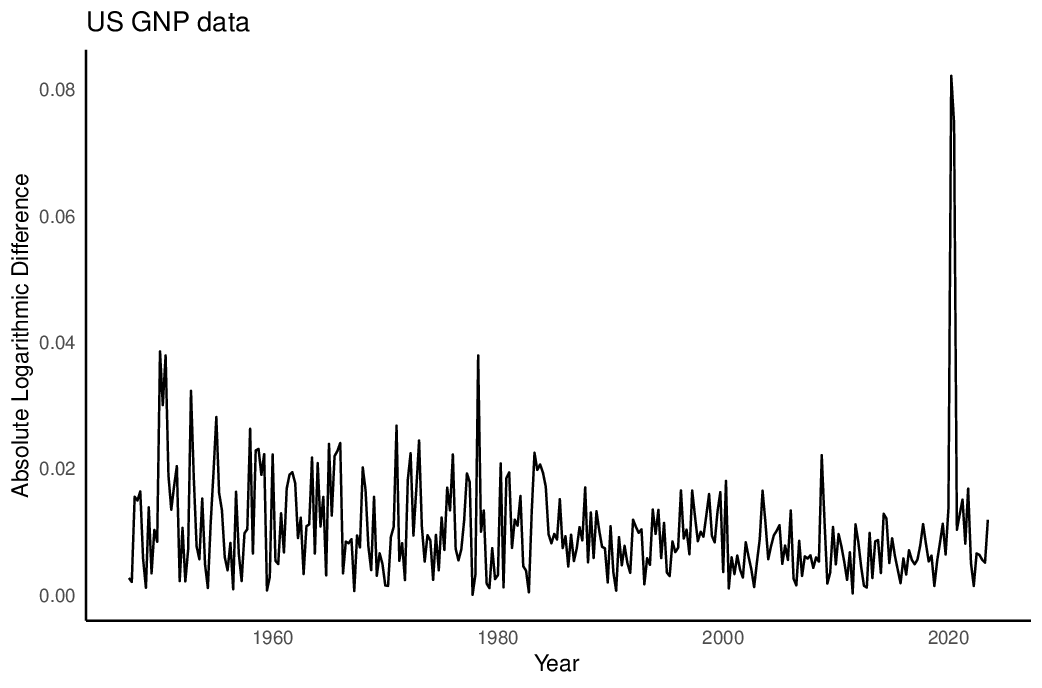}
\label{fig:GNP}
\centering
\end{figure}

\subsection{Argentina rainfall data}

We apply our tests to detect a change in the distribution of annual rainfall data from Argentina. The rainfall data contain annual rainfall measurements, in millimeters, for Argentina from 1884 to 1996. The data set originates from \cite{wu2001isotonic}, where the authors proposed a test statistic based on isotonic regression. The authors in \cite{shao2010testing} considered detecting a change-point in the mean of the data. The authors in \cite{wu2001isotonic} reported that the data provider believed that there was a change in the mean corresponding to the construction of a dam during 1952 -- 1962, and the results from \cite{shao2010testing} supported that hypothesis.

Since we were unable to locate the tabulated data, we used Figure 5 from \cite{shao2010testing} to reconstruct the values in the data set; see Figure \ref{fig:argentina}. We performed single-change-point detection using the permutation bootstrap with $B=1,000$ replications to obtain $p$-values. The results are presented in Table \ref{pvalArgentina}. The proposed tests reported $p$-values less than 0.05, and the estimated change-point location was 1956 for each proposed test. Some competitor tests did not detect the change-point and their location estimates were imprecise. The results of the proposed tests are consistent with those reported in \cite{ shao2010testing, wu2001isotonic}.

\begin{figure}[htbp]
\caption{Argentina rainfall data: yearly rainfall ({millimeters}) in Argentina from 1884 to 1996.}
\includegraphics[width=0.55\textwidth]{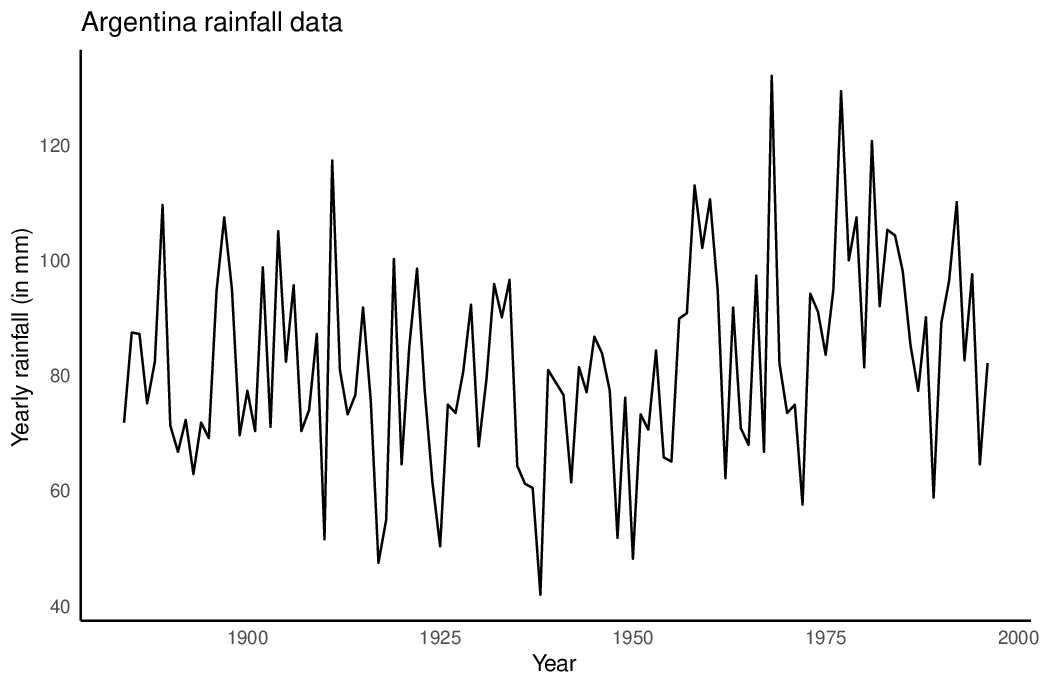}
\label{fig:argentina}
\centering
\end{figure}

\subsection{S\&P 500 Daily Absolute Log-Returns -- 1 January 2018 to
31 December 2024}

The S\&P~500 daily closing prices for the period from 1 January 2018 to 31 December 2024 were downloaded from the Federal Reserve Economic Data (FRED) repository (\href{https://fred.stlouisfed.org}{https://fred.stlouisfed.org}). The period includes several important events that may be of interest for backtesting trading algorithms, including "Volmageddon" \cite{augustin2021volmageddon} and the COVID-19 pandemic, which significantly affected the S\&P 500 index \cite{yilmazkuday2023covid}.
The daily logarithmic returns,
  $r_t = \log\frac{P_t}{P_{t-1}}$ were computed for each trading day, yielding a total of $n = 1{,}759$ observations after removing non-trading dates and missing values.
 Absolute logarithmic returns $|r_t|$ were calculated, and multiple change-points were detected using the binary segmentation algorithm (Algorithm~\ref{BSAl}).

For the S\&P 500 application, we apply the tests to absolute log-returns, which are non-negative and may be interpreted as measuring the magnitude of daily market movements. The analysis is conducted within the IID/exchangeability framework of the proposed tests. Therefore, the detected change-points should be interpreted as evidence of changes in the distribution of return magnitudes under this nonparametric benchmark, rather than as inference from a model explicitly allowing for serial dependence or volatility clustering.

We used $\mathit{Window} = 30$ and $B = 1,000$ bootstrap replicates, at significance level $\alpha = 0.05$. To prevent computationally expensive recursion and for diagnostic purposes, the algorithm runs for a maximum of $50$ iterations and records the depth reached at each detected change-point. The self-normalization parameter $\gamma$ was varied across $\gamma = 0.5$ and $\gamma = 1$. The tuning parameter $a$ was varied over the grid $a \in \{1, 1.5, 2, 3, 4\}$ for each of the four test statistics $T^{(1)}$, $T^{(2)}$, $\mathcal{L}$, and $\mathcal{J}$, yielding forty statistic--parameter combinations in total.

It is important to choose the parameter $\mathit{Window}$ appropriately: choosing $\mathit{Window}$ too large causes the algorithm to miss genuine change-points, while choosing it too small leads to spurious detections, particularly during periods of sustained elevated volatility. Fine-tuning is therefore necessary to achieve optimal performance.

Discrepancies in detection between $\gamma = 0.5$ and $\gamma = 1$ occur around the post-COVID stabilization period, the mid-2020 recovery phase, and the later inflation and Fed repricing period. The tests with $\gamma = 1$ tend to detect upward market movements, while those with $\gamma = 0.5$ tend to detect downward movements in these periods. The detected change-points align closely with major episodes of market-wide uncertainty and repricing in U.S. equities; key events in financial media that provide context for the identified change-point locations are presented in Table \ref{tab:sp500cp} for both $\gamma = 0.5$ and $\gamma = 1$. An illustration of the detected change-points is provided in Figure \ref{fig:sp500cp}.

\begin{figure}[H]
    \centering
    \includegraphics[width=0.95\textwidth]{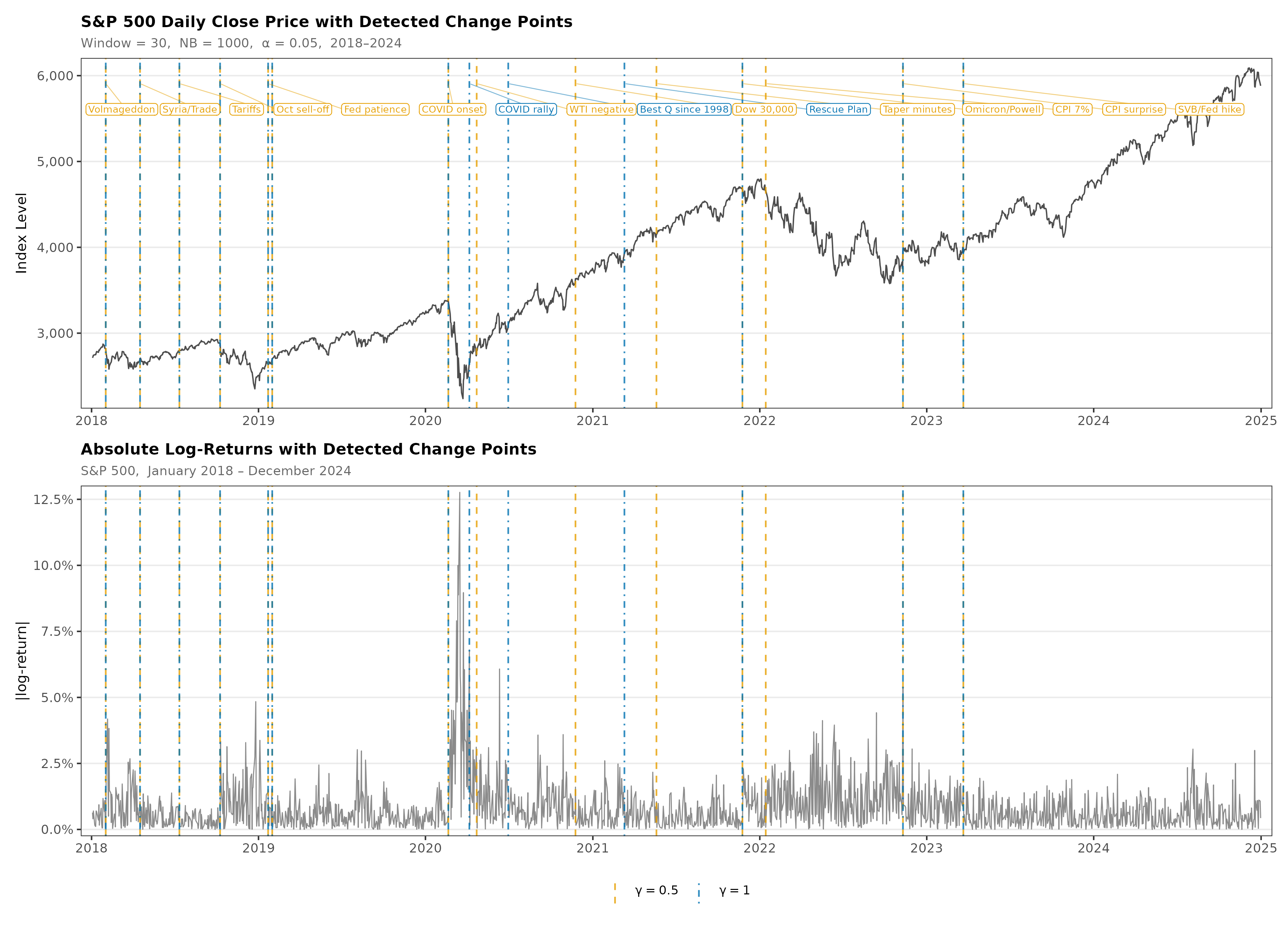}
    \caption{Detected change-points in the S\&P~500 daily close price (top)
             and absolute log-returns (bottom) for $\gamma = 0.5$ and
             $\gamma = 1$. Vertical lines mark the locations detected by
             the binary segmentation procedure with $\mathit{Window} = 30$,
             $NB = 1000$ bootstrap replicates, and significance level
             $\alpha = 0.05$. Amber dashed lines correspond to
             $\gamma = 0.5$; blue dot-dashed lines to $\gamma = 1$.}
    \label{fig:sp500cp}
\end{figure}

\section{Conclusion}\label{sec-conc}

In this paper, we introduced two new classes of nonparametric change-point tests for sequences of univariate non-negative random variables. The proposed procedures are based on the empirical modified Hankel transform and the Laplace transform, respectively, and extend the transform-based approach to change-point detection in a new direction. We derived the asymptotic null distributions of the proposed test statistics and their generalization and discussed the use of a permutation bootstrap procedure for obtaining $p$-values.

The finite-sample study shows that the proposed tests are well calibrated and generally exhibit higher empirical power than the benchmark empirical-characteristic-function-based tests. The simulation results also indicate that the proposed methods provide accurate estimates of the change-point location, particularly in asymmetric settings where the true change-point is not located at the center of the sample. Among the tuning parameters considered, the choice $a=1$ appears to provide stable overall performance across a range of scenarios.

The practical usefulness of the proposed methodology was illustrated through several real-data examples from economics, meteorology, and finance. In the U.S. GNP data, the proposed tests detected a structural break consistent with previous findings in the literature. In the Argentina rainfall data, the tests identified a change-point corresponding to the period associated with the construction of a dam. Finally, the S\&P 500 application demonstrated that the proposed methods can be incorporated into a binary segmentation procedure to detect multiple changes in financial volatility. Overall, these results suggest that the proposed Hankel- and Laplace-transform-based tests provide effective and flexible tools for nonparametric change-point detection in non-negative data.

Future work may consider data-driven selection of the tuning parameter $a$, extensions to dependent observations, and further development of multiple-change-point procedures based on the proposed statistics.

\section*{Disclosure statement}\label{disclosure-statement}

The authors have no interest to declare.

\section*{Data Availability Statement}\label{data-availability-statement}

The data used in this study are available from the authors upon reasonable request.

\section*{Acknowledgment of AI Assistance}
{Claude Sonnet 4.6 (Anthropic, 2025) was used for grammar and language editing, \LaTeX{} typesetting, and code formatting. All scientific content, methodology, analysis, and conclusions were produced independently of any AI tool. No AI tool was used to generate, fabricate, analyze, or impute data. Full responsibility for the accuracy of the work presented is assumed by its contributors.
}
\bibliographystyle{plain}
\bibliography{referencesArXiv}

\newpage
\appendix
\renewcommand{\thesection}{\Alph{section}}
\section{Appendix A}\label{appendixA}
\begin{table}[htbp]
\caption{Model G1: Test powers for different alternatives and $n =40$ using permutation without replacement (PWOR) method and B=500.}
 \label{PWOR}
 \centering
\begin{tabular}{@{}lccccc@{}}\\
\toprule
 &no chpt. & $k=20, b=1$ & $k=20, b=\sqrt2$ & $k=10, b=1$ & $k=10, b=\sqrt2$ \\ \midrule
$\mathcal{L}_{0.5, 0.5}$ & 5 & 98 & 99 & 96 & 98 \\
$\mathcal{L}_{1, 0.5}$ & 5 & 99 & * & 96 & 98 \\
$\mathcal{L}_{0.5, 1}$ & 5 & 98 & 99 & 94 & 97 \\
$\mathcal{L}_{1, 1}$ & 5 & 99 & 99 & 92 & 97 \\
$\mathcal{L}_{0.5, 2}$ & 5 & 97 & 99 & 90 & 95 \\
$\mathcal{L}_{1, 2}$ & 5 & 97 & 99 & 88 & 95 \\
$\mathcal{J}_{0.5, 0.5}$ & 5 & 99  &  95 &  95 & 96  \\
$\mathcal{J}_{0.5, 1}$ & 4 & 98 & * & 93 &  96\\
$\mathcal{J}_{0.5, 2}$ & 5 & 96  &  99 & 88  & 94  \\
$\mathcal{J}_{1, 0.5}$ & 5 &  * &  * &  94 &  96 \\
$\mathcal{J}_{1, 1}$ & 5 & 99 & 99 & 92 & 95 \\
$\mathcal{J}_{1, 2}$ & 5 &  97 &  99 &  87 & 92  \\\midrule
$T^{(1)}_{ 0.5,1}$ & 5 & 96 & 97 & 78 & 78 \\
$T^{(1)}_{ 0.5,1.5}$ & 5 & 94 & 96 & 75 & 77 \\
$T^{(1)}_{ 0.5,2}$ & 5 & 92 & 95 & 73 & 76 \\
$T^{(1)}_{ 0.5,3}$ & 5 & 89 & 94 & 69 & 74 \\
$T^{(1)}_{ 0.5,4}$ & 5 & 87 & 93 & 66 & 72 \\
$T^{(1)}_{ 1,1}$ & 4 & 97 & 98 & 74 & 73 \\
$T^{(1)}_{ 1,1.5}$ & 5 & 96 & 97 & 72 & 73 \\
$T^{(1)}_{ 1,2}$ & 5 & 94 & 96 & 70 & 72 \\
$T^{(1)}_{ 1,3}$ & 5 & 91 & 96 & 66 & 71 \\
$T^{(1)}_{ 1,4}$ & 5 & 89 & 95 & 64 & 70 \\
$T^{(2)}_{ 0.5,1}$ & 5 & 91 & 94 & 72 & 75 \\
$T^{(2)}_{ 0.5,1.5}$ & 5 & 90 & 94 & 69 & 74 \\
$T^{(2)}_{ 0.5,2}$ & 5 & 88 & 93 & 68 & 72 \\
$T^{(2)}_{ 0.5,3}$ & 5 & 86 & 93 & 66 & 71 \\
$T^{(2)}_{ 0.5,4}$ & 5 & 85 & 92 & 64 & 70 \\
$T^{(2)}_{ 1,1}$ & 5 & 93 & 96 & 68 & 71 \\
$T^{(2)}_{ 1,1.5}$ & 5 & 91 & 95 & 66 & 70 \\
$T^{(2)}_{ 1,2}$ & 5 & 90 & 95 & 65 & 70 \\
$T^{(2)}_{ 1,3}$ & 5 & 89 & 95 & 63 & 69 \\
$T^{(2)}_{ 1,4}$ & 5 & 88 & 94 & 62 & 69 \\ \bottomrule
\end{tabular}
\end{table}

 \begin{table}[htbp]
 \caption{Model G2: Test powers for different alternatives and $n =40$ using the PWOR method and B=500.}\label{POWER1}
 \centering
  \centering
\begin{tabular}{@{}lccccc@{}}\\
\toprule
 & no chpt. & $k=20, b=1$ & $k=20, b=\sqrt2$ & $k=10, b=1$ & $k=10, b=\sqrt2$ \\ \midrule
$\mathcal{L}_{0.5, 0.5}$ & 4 & 71 & 92 & 64 & 85 \\
$\mathcal{L}_{1, 0.5}$ & 5 & 78 & 95 & 62 & 83 \\
$\mathcal{L}_{0.5, 1}$ & 4 & 72 & 92 & 62 & 84 \\
$\mathcal{L}_{1, 1}$ & 5 & 76 & 94 & 59 & 82 \\
$\mathcal{L}_{0.5, 2}$ & 5 & 69 & 92 & 58 & 82 \\
$\mathcal{L}_{1, 2}$ & 5 & 74 & 94 & 54 & 79 \\
$\mathcal{J}_{0.5, 0.5}$ & 5 &  69& 87  &  54 &  73  \\
$\mathcal{J}_{0.5, 1}$ & 5 & 71 & 91 & 57 & 77 \\
$\mathcal{J}_{0.5, 2}$ & 5 &  68 & 92  &  55 & 78  \\
$\mathcal{J}_{1, 0.5}$ & 5 &  73 &  90 &  51 &   69\\
$\mathcal{J}_{1, 1}$ & 5 & 74 & 92 & 54 & 73 \\
$\mathcal{J}_{1, 2}$ & 5 & 72  &  93 &  51 &  75 \\\midrule
$T^{(1)}_{0.5, 1}$ & 5 & 52 & 74 & 30 & 45 \\
$T^{(1)}_{0.5, 1.5}$ & 5 & 54 & 79 & 32 & 50 \\
$T^{(1)}_{ 0.5, 2}$ & 5 & 54 & 80 & 32 & 53 \\
$T^{(1)}_{ 0.5, 3}$ & 5 & 54 & 82 & 32 & 55 \\
$T^{(1)}_{ 0.5, 4}$ & 5 & 53 & 83 & 31 & 56 \\
$T^{(1)}_{ 1, 1}$ & 5 & 55 & 77 & 26 & 40 \\
$T^{(1)}_{ 1, 1.5}$ & 5 & 57 & 80 & 28 & 45 \\
$T^{(1)}_{ 1, 2}$ & 5 & 56 & 83 & 28 & 49 \\
$T^{(1)}_{ 1, 3}$ & 5 & 56 & 85 & 29 & 52 \\
$T^{(1)}_{ 1, 4}$ & 5 & 56 & 85 & 29 & 52 \\
$T^{(2)}_{ 0.5, 1}$ & 5 & 53 & 78 & 31 & 50 \\
$T^{(2)}_{ 0.5, 1.5}$ & 6 & 54 & 80 & 31 & 52 \\
$T^{(2)}_{ 0.5, 2}$ & 5 & 53 & 81 & 31 & 54 \\
$T^{(2)}_{ 0.5, 3}$ & 5 & 53 & 82 & 31 & 55 \\
$T^{(2)}_{ 0.5, 4}$ & 5 & 52 & 82 & 31 & 55 \\
$T^{(2)}_{ 1, 1}$ & 5 & 56 & 80 & 28 & 46 \\
$T^{(2)}_{ 1, 1.5}$ & 5 & 56 & 82 & 28 & 49 \\
$T^{(2)}_{ 1,2 }$ & 5 & 56 & 84 & 29 & 50 \\
$T^{(2)}_{ 1, 3}$ & 5 & 55 & 85 & 29 & 51 \\
$T^{(2)}_{ 1, 4}$ & 5 & 55 & 85 & 30 & 53 \\ \bottomrule
\end{tabular}
\end{table}


\begin{table}[htbp]
 \caption{Model G1: Test powers for different alternatives and $n =100$ using the PWOR method and B=500.}
 \label{PWOR100G1}
 \centering
\begin{tabular}{@{}lccccc@{}}\\
\toprule
 & no chpt. & $k=50, b=1$ & $k=50, b=\sqrt2$ & $k=25, b=1$ & $k=25, b=\sqrt2$ \\ \midrule
$\mathcal{L}_{0.5, 0.5}$ & 5 & * & * & * & * \\
$\mathcal{L}_{1, 0.5}$ & 5 & * & * & * & * \\
$\mathcal{L}_{0.5, 1}$ & 5 & * & * & * & * \\
$\mathcal{L}_{1, 1}$ & 5 & * & * & * & * \\
$\mathcal{L}_{0.5, 2}$ & 5 & * & * & * & * \\
$\mathcal{L}_{1, 2}$ & 5 & * & * & * & * \\
$\mathcal{J}_{0.5, 0.5}$ & 4 &  * & *  & *  & * \\
$\mathcal{J}_{0.5, 1}$ & 4 & * & * & * & * \\
$\mathcal{J}_{0.5, 2}$ & 5 &  * &  * &  * &  * \\
$\mathcal{J}_{1, 0.5}$ & 5 & *  & *  & *  &  * \\
$\mathcal{J}_{1, 1}$ & 4 & * & * & * & * \\ 
$\mathcal{J}_{1, 2}$ & 4 &  * & *  &  * &  * \\\midrule
$T^{(1)}_{ 0.5, 1}$ & 5 & * & * & * & 99 \\
$T^{(1)}_{ 0.5, 1.5}$ & 4 & * & * & * & 99 \\
$T^{(1)}_{ 0.5, 2}$ & 4 & * & * & 99 & 99 \\
$T^{(1)}_{ 0.5, 3}$ & 5 & * & * & 99 & 99 \\
$T^{(1)}_{ 0.5, 4}$ & 5 & * & * & 98 & 99 \\
$T^{(1)}_{ 1, 1}$ & 5 & * & * & * & 99 \\
$T^{(1)}_{ 1, 1.5}$ & 5 & * & * & 99 & 99 \\
$T^{(1)}_{ 1, 2}$ & 5 & * & * & 99 & 99 \\
$T^{(1)}_{ 1, 3}$ & 5 & * & * & 98 & 99 \\
$T^{(1)}_{ 1, 4}$ & 5 & * & * & 97 & 99 \\
$T^{(2)}_{ 0.5, 1}$ & 5 & * & * & 99 & 99 \\
$T^{(2)}_{ 0.5, 1.5}$ & 5 & * & * & 99 & 99 \\
$T^{(2)}_{ 0.5, 2}$ & 5 & * & * & 98 & 99 \\
$T^{(2)}_{ 0.5, 3}$ & 5 & * & * & 98 & 99 \\
$T^{(2)}_{ 0.5, 4}$ & 5 & * & * & 97 & 99 \\
$T^{(2)}_{ 1, 1}$ & 5 & * & * & 99 & 99 \\
$T^{(2)}_{ 1, 1.5}$ & 5 & * & * & 98 & 99 \\
$T^{(2)}_{ 1, 2}$ & 5 & * & * & 98 & 99 \\
$T^{(2)}_{ 1, 3}$ & 5 & * & * & 97 & 99 \\
$T^{(2)}_{ 1, 4}$ & 5 & * & * & 97 & 99 \\ \bottomrule
\end{tabular}
\end{table}

\begin{table}[htbp]
 \caption{Model G2: Test powers for different alternatives and $n =100$ using the PWOR method and B=500.} 
  \label{PWOR100G2}
 \centering
\begin{tabular}{@{}lccccc@{}}\\
\toprule
 & no chpt. & $k=50, b=1$ & $k=50, b=\sqrt2$ & $k=25, b=1$ & $k=25, b=\sqrt2$ \\ \midrule
$\mathcal{L}_{0.5, 0.5}$ & 5 & 99 & * & 97 & * \\
$\mathcal{L}_{1, 0.5}$ & 4 & * & * & 96 & * \\
$\mathcal{L}_{0.5, 1}$ & 5 & 99 & * & 96 & * \\
$\mathcal{L}_{1, 1}$ & 5 & * & * & 95 & * \\
$\mathcal{L}_{0.5, 2}$ & 5 & 99 & * & 95 & * \\
$\mathcal{L}_{1, 2}$ & 5 & 99 & * & 93 & * \\
$\mathcal{J}_{0.5, 0.5}$ & 5 & *  & *  & 93   & 99  \\
$\mathcal{J}_{0.5, 1}$ & 5 & 99 & * & 94 & * \\
$\mathcal{J}_{0.5, 2}$ & 5 & 99  &  * &  93 &  99 \\
$\mathcal{J}_{1, 0.5}$ & 5 &  99 & *  & 91  &  99 \\
$\mathcal{J}_{1, 1}$ & 5 & 99 & * & 93 & * \\
$\mathcal{J}_{1, 2}$ & 5 &  99 & *  &  92  &  99 \\\midrule
$T^{(1)}_{ 0.5, 1}$ & 5 & 94 & 99 & 74 & 92 \\
$T^{(1)}_{ 0.5, 1.5}$ & 5 & 94 & 99 & 76 & 94 \\
$T^{(1)}_{ 0.5, 2}$ & 6 & 94 & 99 & 77 & 95 \\
$T^{(1)}_{ 0.5, 3}$ & 6 & 93 & 99 & 76 & 96 \\
$T^{(1)}_{ 0.5, 4}$ & 6 & 92 & 99 & 75 & 97 \\
$T^{(1)}_{ 1, 1}$ & 5 & 95 & 99 & 69 & 90 \\
$T^{(1)}_{ 1, 1.5}$ & 5 & 95 & 99 & 71 & 93 \\
$T^{(1)}_{ 1, 2}$ & 5 & 95 & 99 & 72 & 94 \\
$T^{(1)}_{ 1, 3}$ & 5 & 95 & 99 & 73 & 95 \\
$T^{(1)}_{ 1, 4}$ & 5 & 94 & * & 72 & 96 \\
$T^{(2)}_{ 0.5, 1}$ & 5 & 93 & 99 & 76 & 94 \\
$T^{(2)}_{ 0.5, 1.5}$ & 6 & 93 & 99 & 76 & 95 \\
$T^{(2)}_{ 0.5, 2}$ & 6 & 93 & 99 & 76 & 96 \\
$T^{(2)}_{ 0.5, 3}$ & 6 & 92 & 99 & 75 & 96 \\
$T^{(2)}_{ 0.5, 4}$ & 6 & 92 & * & 74 & 97 \\
$T^{(2)}_{ 1, 1}$ & 5 & 95 & 99 & 72 & 93 \\
$T^{(2)}_{ 1, 1.5}$ & 5 & 95 & 99 & 72 & 94 \\
$T^{(2)}_{ 1, 2}$ & 5 & 94 & 99 & 72 & 95 \\
$T^{(2)}_{ 1, 3}$ & 6 & 94 & 99 & 72 & 96 \\
$T^{(2)}_{ 1, 4}$ & 5 & 93 & * & 71 & 96 \\ \bottomrule
\end{tabular}
\end{table}



\begin{table}[htbp]
\caption{Estimates and standard deviations of change-point locations, $n = 40$.}
\label{Gamma05Pozicija40}
\centering
\begin{tabular}{@{}lccc@{}@{\hspace{5pt}}|@{\hspace{2pt}}ccc@{}}\\
\toprule \centering
& \multicolumn{3}{c}{$b=1, k=20$}& \multicolumn{3}{c}{$b=1, k=10$}\\\midrule
 & G1 & G2 & U &   G1 & G2 & U  \\\midrule
$\mathcal{L}_{0.5, 0.5}$ & 18.59 (2.405) & 17.59 (4.344) & 19.55 (0.914)& 9.4 (1.932) & 10.23 (5.123) & 9.85 (0.485) \\

$\mathcal{L}_{1, 0.5}$ & 18.85 (1.924) & 18.23 (3.495) & 19.6 (0.803)&10.04 (1.987)&11.48 (4.882) &10.03 (0.256) \\
$\mathcal{L}_{0.5, 1}$ & 18.77 (2.305) & 17.91 (4.374) & 19.71 (0.672)&9.68 (2.214) & 10.7 (5.471) & 9.95 (0.289) \\
$\mathcal{L}_{1, 1}$ & 18.99 (1.871) & 18.43 (3.508) & 19.74 (0.601)&10.39 (2.35) & 11.77 (5.015) & 10.09 (0.357) \\
$\mathcal{L}_{0.5, 2}$ & 18.98 (2.355) & 18.25 (4.476) & 19.83 (0.476)& 10.1 (2.648) & 11.18 (5.779) & 9.99 (0.221)  \\
$\mathcal{L}_{1, 2}$ & 19.17 (1.938) & 18.64 (3.616) & 19.85 (0.434) &10.92 (2.873) & 12.28 (5.331) & 10.19 (0.529)  \\
$\mathcal{J}_{0.5, 0.5}$ & 18.73 (2.212)  & 18.06 (4.472) & 19.69 (0.720) & 9.52 (2.155) & 11.36 (6.034) & 9.92 (0.376)  \\
$\mathcal{J}_{1, 0.5}$ & 18.93  (1.811)& 18.51 (3.584)& 19.73 (0.633) & 10.09 (2.094) & 12.45 (5.522) &  10.06 (0.344) \\
$\mathcal{J}_{0.5, 1}$ & 18.87 (2.204) & 18.12 (4.442) & 19.78 (0.567) &9.75 (2.169) & 11.29 (5.925) & 9.97 (0.25)  \\

$\mathcal{J}_{1, 1}$ & 19.06 (1.756) & 18.57 (3.610) & 19.8 (0.501)& 10.39 (2.377) & 12.29 (5.37) & 10.12 (0.43)  \\
$\mathcal{J}_{0.5, 2}$ & 19.09 (2.388)   & 18.39 (4.567) & 19.89 (0.372) & 10.27 (2.900) & 11.44 (5.926)  &  10.03 (0.269) \\
$\mathcal{J}_{1, 2}$ & 19.24 (1.993) & 18.74 (3.607) & 19.91 (0.338) & 11.08 (3.028) & 12.56 (5.525)  &  10.25 (0.676) \\
\midrule
$T^{(1)}_{0.5, 1}$ & 19.27 (2.13) & 19.3 (5.083) & 20.01 (0.224)& 10.97 (4.085) & 14.38 (7.485) & 10.13 (0.513)  \\
$T^{(1)}_{0.5, 1.5}$ & 19.36 (2.356) & 19.39 (5.187) & 20.01 (0.223)& 11.2 (4.395) & 14.03 (7.345) & 10.15 (0.538) \\
$T^{(1)}_{0.5, 2}$ & 19.48 (2.524) & 19.47 (5.174) & 20.01 (0.215)& 11.42 (4.583) & 13.98 (7.434) & 10.16 (0.545)  \\
$T^{(1)}_{0.5, 3}$ & 19.464 (2.873) & 19.513 (5.312) & 20 (0.198) & 11.626 (4.88) & 14.175 (7.643) & 10.174 (0.53) \\ 
$T^{(1)}_{0.5, 4}$ & 19.52 (3.137) & 19.541 (5.493) & 20 (0.192) & 11.868 (5.186) & 14.217 (7.753) & 10.18 (0.537) \\

$T^{(1)}_{1, 1}$ & 19.38 (1.729) & 19.34 (4.037) & 20.01 (0.206) & 11.85 (4.229) & 15.25 (6.404) & 10.33 (0.827)\\
$T^{(1)}_{1, 1.5}$ & 19.48 (1.867) & 19.38 (4.019) & 20.01 (0.195) & 11.99 (4.259) & 14.89 (6.401) & 10.38 (0.86) \\

$T^{(1)}_{1, 2}$ & 19.52 (2.061) & 19.44 (4.141) & 20.01 (0.196)  & 12.23 (4.407) & 14.74 (6.428) & 10.41 (0.879)  \\

$T^{(1)}_{1, 3}$ & 19.538 (2.357)& 19.577 (4.177)& 20 (0.171)& 12.403 (4.511) & 14.854 (6.621) & 10.446 (0.918)\\
$T^{(1)}_{1, 4}$ & 19.574 (2.524) & 19.609 (4.29) & 20 (0.169) & 12.655 (4.71) & 14.866 (6.683) & 10.461 (0.936) \\

$T^{(2)}_{0.5, 1}$ & 19.52 (2.65) & 19.42 (5.358) & 20.01 (0.217) & 11.56 (4.688) & 14.12 (7.527) & 10.17 (0.556) \\
$T^{(2)}_{0.5, 1.5}$ & 19.57 (2.841) & 19.45 (5.366) & 20.01 (0.212)& 11.87 (5.031) & 14 (7.543) & 10.18 (0.577) \\

$T^{(2)}_{0.5, 2}$ & 19.61 (2.962) & 19.5 (5.473) & 20.01 (0.211)& 12.07 (5.307) & 14.03 (7.623) & 10.19 (0.586) \\
$T^{(2)}_{0.5, 3}$ & 19.972 (2.351) & 20.189 (2.911) & 20 (0) & 12.964 (4.658) & 13.844 (5.5) & 10.297 (0.706)\\
$T^{(2)}_{0.5, 4}$ & 20.033 (2.383) & 20.24 (2.916) & 20 (0) & 13.066 (4.724) & 13.848 (5.455) & 10.308 (0.718)\\

$T^{(2)}_{1, 1}$ & 19.57 (2.134) & 19.51 (4.205) & 20.01 (0.186)& 12.39 (4.497) & 14.82 (6.529) & 10.42 (0.893) \\

$T^{(2)}_{1, 1.5}$ & 19.65 (2.301) & 19.55 (4.362) & 20.01 (0.186)& 12.57 (4.661) & 14.73 (6.491) & 10.44 (0.905) \\ 
$T^{(2)}_{1, 2}$ & 19.67 (2.485) & 19.53 (4.343) & 20.01 (0.184) & 12.73 (4.867) & 14.69 (6.494) & 10.45 (0.906) \\
$T^{(2)}_{1, 3}$ & 19.602 (2.574) & 19.607 (4.335) & 20 (0.169) & 12.73 (4.777) & 14.892 (6.707) & 10.471 (0.945) \\
$T^{(2)}_{1, 4}$ & 19.623 (2.67) & 19.618 (4.378) & 20 (0.168) & 12.85 (4.878) & 14.905 (6.745) & 10.477 (0.951) \\

 \bottomrule
\end{tabular}
\end{table}

\begin{table}[htbp]\ContinuedFloat
\caption{Estimates and standard deviations of change-point locations{, $n=40$ (continued)}.}
\centering
\begin{tabular}{{@{}lccc@{}@{\hspace{5pt}}|@{\hspace{2pt}}ccc@{}}}\\
\toprule
& \multicolumn{3}{c}{$b=\sqrt{2}, k=20$} &  \multicolumn{3}{c}{$b=\sqrt{2}, k=10$} \\\midrule
 & G1 & G2 & U  & G1 & G2 & U \\ \midrule
$\mathcal{L}_{0.5, 0.5}$ & 18.7 (2.31) & 18.57 (3.141) & 19.81 (0.572) & 9.47 (1.619) & 10.08 (3.457) & 9.98 (0.158) \\
$\mathcal{L}_{1, 0.5}$ & 18.94 (1.841) & 18.92 (2.458) & 19.84 (0.484) & 10.1 (1.765) & 11.04 (3.347) & 10 (0.05) \\
$\mathcal{L}_{0.5, 1}$ & 18.93 (2.206) & 18.87 (3.037) & 19.94 (0.296) & 9.79 (1.895) & 10.34 (3.522) & 10 (0.022) \\
$\mathcal{L}_{1, 1}$ & 19.12 (1.748) & 19.11 (2.424) & 19.95 (0.246)  & 10.44 (2.096) & 11.35 (3.527) & 10.02 (0.144) \\
$\mathcal{L}_{0.5, 2}$ & 19.16 (2.128) & 19.15 (2.947) & 19.99 (0.109)  & 10.16 (2.25) & 10.73 (3.696) & 10 (0) \\
$\mathcal{L}_{1, 2}$ & 19.32 (1.785) & 19.36 (2.389) & 19.99 (0.08)  & 10.88 (2.581) & 11.75 (3.776) & 10.05 (0.23) \\
$\mathcal{J}_{0.5, 0.5}$ & 18.74 (2.314)  &18.99 (3.374)& 19.94 (0.301)&9.58 (1.924)& 10.96 (4.627)&  9.99 (0.083)\\
$\mathcal{J}_{1, 0.5}$ & 18.99 (1.832)  &19.20 (2.672)& 19.96 (0.246)&10.14 (2.116)& 12.04 (4.455)&  10.01 (0.118)\\
$\mathcal{J}_{0.5, 1}$ & 18.97 (2.211) & 19.05 (3.172) & 19.98 (0.172)  & 9.86 (2.046) & 10.9 (4.310) & 10 (0) \\
$\mathcal{J}_{1, 1}$ & 19.18 (1.728) & 19.29 (2.542) & 19.98 (0.149)  & 10.46 (2.291) & 11.9 (4.145) & 10.03 (0.168) \\
$\mathcal{J}_{0.5, 2}$ &  19.25 (2.198) &19.29 (3.015)& 20.00 (0.039)&10.32 (2.438)& 11.08 (4.201)&  10.00 (0.022)\\
$\mathcal{J}_{1, 2}$ & 19.41 (1.792)  &19.45 (2.443)& 20.00 (0.032)& 11.04 (2.762)& 12.02 (4.033)&  10.08 (0.312)\\ \midrule
$T^{(1)}_{0.5, 1}$ & 19.39 (2.321) & 19.73 (3.746) & 20 (0) & 11.04 (3.793) & 13.51 (6.22) & 10.02 (0.158) \\
$T^{(1)}_{0.5, 1.5}$ & 19.56 (2.44) & 19.81 (3.68) & 20 (0.022) & 11.27 (4.065) & 13.3 (6.202) & 10.04 (0.197) \\
$T^{(1)}_{0.5, 2}$ & 19.67 (2.529) & 19.97 (3.716) & 20 (0.022) & 11.52 (4.367) & 13.28 (6.156) & 10.05 (0.256) \\
$T^{(1)}_{0.5, 3}$ & 19.864 (2.663) & 20.167 (3.631) & 20 (0.014) & 12.003 (4.866) & 13.137 (6.006) & 10.059 (0.277)\\
$T^{(1)}_{0.5, 4}$ & 19.988 (2.813) & 20.281 (3.585) & 20 (0.014) & 12.246 (5.061) & 13.188 (6.034) & 10.073 (0.316)\\
$T^{(1)}_{1, 1}$ & 19.46 (1.875) & 19.76 (2.991) & 20 (0) & 12.01 (4.072) & 14.54 (5.742) & 10.12 (0.438) \\
$T^{(1)}_{1, 1.5}$ & 19.57 (2.054) & 19.81 (2.88) & 20 (0.022) & 12.15 (4.198) & 14.25 (5.689) & 10.17 (0.529) \\
$T^{(1)}_{1, 2}$ & 19.69 (2.144) & 19.95 (2.869) & 20 (0.022) & 12.22 (4.154) & 13.96 (5.509) & 10.22 (0.608) \\
$T^{(1)}_{1, 3}$ & 19.867 (2.214) & 20.107 (2.883) & 20 (0) & 12.719 (4.501) & 13.793 (5.459) & 10.255 (0.65) \\
$T^{(1)}_{1, 4}$ & 19.961 (2.316) & 20.195 (2.863) & 20 (0) & 12.924 (4.635) & 13.8 (5.449) & 10.283 (0.684)\\ 
$T^{(2)}_{0.5, 1}$ & 19.72 (2.664) & 19.91 (3.731) & 20 (0.022) & 11.72 (4.594) & 13.29 (6.214) & 10.06 (0.278) \\
$T^{(2)}_{0.5, 1.5}$ & 19.8 (2.838) & 20.05 (3.738) & 20 (0.022)  & 11.83 (4.687) & 13.28 (6.151) & 10.07 (0.297) \\
$T^{(2)}_{0.5, 2}$ & 19.9 (2.885) & 20.15 (3.68) & 20 (0.022) & 12.02 (4.859) & 13.3 (6.164) & 10.07 (0.304) \\
$T^{(2)}_{0.5, 3}$ & 19.545 (3.226) & 19.535 (5.513) & 20 (0.191) & 11.925 (5.266) & 14.219 (7.739) & 10.183 (0.543) \\ 
$T^{(2)}_{0.5, 4}$ & 19.574 (3.363) & 19.56 (5.625) & 20 (0.191) & 12.073 (5.45) & 14.249 (7.813) & 10.185 (0.545) \\
$T^{(2)}_{1, 1}$ & 19.71 (2.211) & 19.92 (3.019) & 20 (0.022)  & 12.38 (4.332) & 14.18 (5.702) & 10.25 (0.639) \\
$T^{(2)}_{1, 1.5}$ & 19.81 (2.290) & 20.02 (2.933) & 20 (0.022) & 12.54 (4.421) & 14.05 (5.635) & 10.28 (0.668) \\
$T^{(2)}_{1, 2}$ & 19.87 (2.296) & 20.04 (2.933) & 20 (0.022) & 12.59 (4.427) & 14.03 (5.592) & 10.32 (0.714) \\
$T^{(2)}_{1, 3}$ & 20.005 (2.851) & 20.293 (3.651) & 20 (0.014) & 12.33 (5.14) & 13.227 (6.07) & 10.077 (0.326)\\
$T^{(2)}_{1, 4}$ & 20.072 (2.92) & 20.35 (3.653) & 20 (0.017) & 12.433 (5.218) & 13.258 (6.101) & 10.081 (0.332)\\
\bottomrule
\end{tabular}
\end{table}

\begin{table}[htbp]
\caption{Estimates and standard deviations of change-point locations{, $n=100$}.}
\label{Gamma05Pozicija100}
\centering
\begin{tabular}{{@{}lccc@{}@{\hspace{5pt}}|@{\hspace{5pt}}ccc@{}}}\\
\toprule
& \multicolumn{3}{c}{$b=1, k=50$} & \multicolumn{3}{c}{$b=1, k=25$}\\\midrule
 & G1 & G2 & U  & G1 & G2 & U \\ \midrule
$\mathcal{L}_{0.5, 0.5}$ & 48.42 (2.792) & 46.98 (5.494) & 49.54 (0.977) & 24.17 (1.909) & 24.33 (5.44) & 24.83 (0.499) \\
$\mathcal{L}_{1, 0.5}$ & 48.6 (2.414) & 47.46 (4.476) & 49.57 (0.896) & 24.83 (1.88) & 26.44 (5.736) & 25.03 (0.236) \\
$\mathcal{L}_{0.5, 1}$ & 48.62 (2.644) & 47.31 (5.561) & 49.72 (0.674)  & 24.41 (2.077) & 24.8 (5.728) & 24.92 (0.299) \\
$\mathcal{L}_{1, 1}$ & 48.75 (2.336) & 47.77 (4.47) & 49.73 (0.643) & 25.21 (2.247) & 27.12 (6.456) & 25.09 (0.337) \\
$\mathcal{L}_{0.5, 2}$ & 48.84 (2.584) & 47.62 (5.706) & 49.83 (0.49) & 24.76 (2.464) & 25.74 (6.886) & 24.99 (0.202) \\
$\mathcal{L}_{1, 2}$ & 48.91 (2.404) & 48.17 (4.546) & 49.83 (0.474) & 25.85 (3.11) & 27.9 (7.185) & 25.17 (0.47) \\
$\mathcal{J}_{0.5, 0.5}$ &  48.56 (2.592) &47.44 (5.679)& 49.72 (0.683)&24.25 (1.819)& 25.41 (7.062)& 24.90 (0.356)  \\
$\mathcal{J}_{1, 0.5}$ &  48.69 (2.292)&47.85  (4.639)& 49.73 (0.652)& 24.77 (1.746)& 27.77 (7.950)&25.06 (0.281) \\
$\mathcal{J}_{0.5, 1}$ & 48.7 (2.518) & 47.47 (5.671) & 49.78 (0.584) & 24.46 (1.969) & 25.34 (6.668) & 24.96 (0.250) \\
$\mathcal{J}_{1, 1}$ & 48.83 (2.234) & 47.98 (4.552) & 49.79 (0.543) & 25.17 (2.134) & 27.72 (7.331) & 25.12 (0.388) \\ 
$\mathcal{J}_{0.5, 2}$ &  48.94 (2.577)&47.75 (5.799)& 49.89 (0.372)& 24.90 (2.580)& 26.18 (7.416)& 25.03 (0.220) \\
$\mathcal{J}_{1, 2}$ & 49.00 (2.441) & 48.23 (4.619) & 49.89 (0.364)& 26.17 (3.635)& 28.21 (7.590)& 25.23 (0.600) \\ \midrule
$T^{(1)}_{ 0.5, 1}$ & 49.07 (2.208) & 48.41 (6.485) & 50.01 (0.156)  & 25.02 (3.155) & 29.35 (11.845) & 25.11 (0.384) \\
$T^{(1)}_{ 0.5, 1.5}$ & 49.16 (2.401) & 48.63 (6.707) & 50.01 (0.141) & 25.32 (3.694) & 29.26 (11.634) & 25.12 (0.41) \\
$T^{(1)}_{ 0.5, 2}$ & 49.25 (2.631) & 48.72 (6.733) & 50.01 (0.139)  & 25.61 (4.057) & 29.12 (11.325) & 25.13 (0.422) \\

$T^{(1)}_{ 0.5, 3}$ & 49.16 (3.254)  & 49.13 (7.011)  &  50 (0.159)  &  25.88 (4.433) & 29.24 (11.379)  & 25.16 (0.483)  \\
$T^{(1)}_{ 0.5, 4}$ & 49.24 (3.512)  & 49.28 (7.213)  &  50 (0.158)  &  26.23 (5.024) & 29.43 (11.661)  & 25.17 (0.505)   \\

$T^{(1)}_{ 1, 1}$ & 49.14 (2.028) & 48.61 (5.066) & 50.01 (0.145) & 25.88 (3.619) & 31.82 (11.452) & 25.32 (0.808) \\
$T^{(1)}_{ 1, 1.5}$ & 49.21 (2.214) & 48.7 (5.125) & 50.01 (0.138) & 26.46 (4.365) & 31.43 (11.113) & 25.36 (0.862) \\
$T^{(1)}_{ 1, 2}$ & 49.3 (2.42) & 48.87 (5.218) & 50.01 (0.138) & 26.88 (4.941) & 31.57 (11.11) & 25.42 (0.952) \\

$T^{(1)}_{ 1, 3}$ & 49.27 (2.834)  & 49.26 (5.546)  & 50 (0.153)   & 27.42 (5.454)  & 31.21 (10.693) & 25.47 (1)  \\
$T^{(1)}_{ 1, 4}$ &  49.34 (3.044) & 49.35 (5.748)  &  50 (0.151)  & 27.83 (5.873)  &  31.46 (10.955) &  25.49 (1.027) \\

$T^{(2)}_{ 0.5, 1}$ & 49.27 (2.843) & 48.77 (7.04) & 50.01 (0.139) & 25.75 (4.266) & 29.28 (11.505) & 25.14 (0.433) \\
$T^{(2)}_{ 0.5, 1.5}$ & 49.32 (3.039) & 48.87 (7.045) & 50.01 (0.139) & 26.1 (4.871) & 29.58 (11.741) & 25.15 (0.436) \\
$T^{(2)}_{ 0.5, 2}$ & 49.41 (3.21) & 49.01 (7.284) & 50.01 (0.139) & 26.3 (5.218) & 29.73 (11.839) & 25.15 (0.445) \\

$T^{(2)}_{ 0.5, 3}$ & 49.26 (3.583)  &  49.32 (7.307) &  50 (0.158)  &  26.35 (5.279) & 29.53 (11.770)  &   25.17 (0.509)\\
$T^{(2)}_{ 0.5, 4}$ & 49.30 (3.727)  & 49.31 (7.440)  & 50 (0.157)   &  26.53 (5.626) & 29.61 (11.790)  & 25.17 (0.513)   \\

$T^{(2)}_{ 1, 1}$ & 49.35 (2.525) & 48.87 (5.391) & 50.01 (0.13) & 27.21 (5.338) & 31.62 (11.235) & 25.45 (0.979) \\
$T^{(2)}_{ 1, 1.5}$ & 49.4 (2.743) & 48.95 (5.697) & 50.01 (0.128) & 27.66 (5.817) & 31.74 (11.401) & 25.47 (1.008) \\
$T^{(2)}_{ 1, 2}$ & 49.46 (2.837) & 49.08 (5.778) & 50.01 (0.128) & 27.92 (6.128) & 31.83 (11.592) & 25.48 (1.024) \\ 

$T^{(2)}_{ 1, 3}$ & 49.35 (3.124)  &  49.37 (5.823) &  50 (0.150)  & 27.98 (6.066)  &  31.54 (11.055) & 25.5 (1.045)  \\
$T^{(2)}_{ 1, 4}$ &  49.38 (3.231) & 49.4 (5.972)  & 50 (0.15)   & 28.23 (6.386)  & 31.71 (11.218)  &  25.51 (1.052) \\

\bottomrule

\end{tabular}
\end{table}

\begin{table}[htbp]\ContinuedFloat
\caption{Estimates and standard deviations of change-point locations{, $n=100$ (continued)}.}
\centering
\begin{tabular}{{@{}lccc@{}@{\hspace{5pt}}|@{\hspace{5pt}}ccc@{}}}\\
\toprule
& \multicolumn{3}{c}{$b=\sqrt{2}, k=50$} & \multicolumn{3}{c}{$b=\sqrt{2}, k=25$} \\\midrule
 & G1 & G2 & U  & G1 & G2 & U  \\ \midrule
$\mathcal{L}_{0.5, 0.5}$ & 48.55 (2.586) & 48.25 (3.725) & 49.84 (0.496) & 24.27 (1.825) & 24.64 (3.197) & 24.99 (0.077) \\
$\mathcal{L}_{1, 0.5}$ & 48.69 (2.259) & 48.48 (3.124) & 49.85 (0.435)  & 24.92 (1.768) & 26.12 (3.815) & 25 (0.032) \\
$\mathcal{L}_{0.5, 1}$ & 48.83 (2.373) & 48.56 (3.531) & 49.95 (0.257) & 24.56 (1.853) & 24.99 (3.298) & 25 (0) \\
$\mathcal{L}_{1, 1}$ & 48.97 (2.115) & 48.75 (3.034) & 49.96 (0.24) & 25.3 (2.193) & 26.53 (4.288) & 25.01 (0.095) \\
$\mathcal{L}_{0.5, 2}$ & 49.1 (2.299) & 48.91 (3.414) & 50 (0.045)  & 24.87 (2.136) & 25.53 (3.781) & 25 (0) \\
$\mathcal{L}_{1, 2}$ & 49.19 (2.109) & 49.03 (2.952) & 50 (0.039) & 25.84 (2.858) & 27.07 (4.773) & 25.05 (0.229) \\
$\mathcal{J}_{0.5, 0.5}$ & 48.62 (2.530) & 48.65 (3.932)& 49.97 (0.219)& 24.30 (1.866)& 25.41 (4.319)& 25 (0.022) \\
$\mathcal{J}_{1, 0.5}$ &  48.74 (2.211)& 48.86 (3.246)& 49.97 (0.184)& 24.88 (1.763)& 27.00 (5.238)&25.01 (0.077)  \\
$\mathcal{J}_{0.5, 1}$ & 48.89 (2.322) & 48.76 (3.532) & 49.99 (0.074) & 24.58 (1.963) & 25.42 (3.877) & 25 (0) \\
$\mathcal{J}_{1, 1}$ & 48.97 (2.143) & 48.9 (3.210) & 49.99 (0.074)  & 25.24 (2.187) & 26.91 (4.799) & 25.01 (0.118) \\ 
$\mathcal{J}_{0.5, 2}$ & 49.16 (2.337) & 48.99 (3.467)& 50 (0)& 25 (2.345)&25.77 (4.070) &25 (0)  \\
$\mathcal{J}_{1, 2}$ & 49.24 (2.118) &49.13 (3.124) & 50 (0)& 25.98 (3.047)& 27.30 (5.003) & 25.08 (0.308) \\\midrule
$T^{(1)}_{ 0.5, 1}$ & 49.05 (2.427) & 49.35 (4.224) & 50 (0) & 25.04 (2.855) & 27.44 (7.333) & 25.01 (0.111) \\
$T^{(1)}_{ 0.5, 1.5}$ & 49.24 (2.484) & 49.56 (4.103) & 50 (0) & 25.31 (3.261) & 27.47 (7.043) & 25.02 (0.158) \\
$T^{(1)}_{ 0.5, 2}$ & 49.37 (2.598) & 49.64 (4.17) & 50 (0)  & 25.6 (3.671) & 27.61 (7.106) & 25.04 (0.2) \\

$T^{(1)}_{ 0.5, 3}$ & 49.71 (2.757) & 49.91 (4.095)  &  50 (0)  & 26.3 (4.619)  & 27.67 (7.200)  & 25.06 (0.265) \\
$T^{(1)}_{ 0.5, 4}$ & 49.89 (2.896) & 50.12 (4.054) &  50 (0)  & 26.64 (5.023)  & 27.84 (7.194)  & 25.07 (0.293) \\

$T^{(1)}_{ 1, 1}$ & 49.11 (2.266) & 49.42 (3.529) & 50 (0) & 26.05 (3.784) & 29.67 (8.305) & 25.12 (0.403) \\
$T^{(1)}_{ 1, 1.5}$ & 49.32 (2.244) & 49.55 (3.517) & 50 (0) & 26.59 (4.502) & 29.52 (7.876) & 25.17 (0.492) \\
$T^{(1)}_{ 1, 2}$ & 49.45 (2.273) & 49.67 (3.502) & 50 (0) & 27.05 (5.048) & 29.46 (7.623) & 25.2 (0.552) \\
$T^{(1)}_{ 1, 3}$ & 49.76 (2.426) & 49.94 (3.597)  &  50 (0)  &  27.82 (5.644) & 29.37 (7.718)  & 25.26 (0.654)  \\
$T^{(1)}_{ 1, 4}$ & 49.89 (2.561) & 50.09 (3.569) &  50 (0)  & 28.22 (5.999)  & 29.57 (7.779)  & 25.29 (0.703) \\

$T^{(2)}_{ 0.5, 1}$ & 49.39 (2.778) & 49.65 (4.346) & 50 (0) & 25.77 (4.02) & 27.7 (7.422) & 25.05 (0.227) \\
$T^{(2)}_{ 0.5, 1.5}$ & 49.56 (2.774) & 49.75 (4.25) & 50 (0) & 25.99 (4.162) & 27.86 (7.365) & 25.06 (0.263) \\
$T^{(2)}_{ 0.5, 2}$ & 49.68 (2.729) & 49.83 (4.129) & 50 (0) & 26.23 (4.739) & 27.81 (7.194) & 25.06 (0.272) \\

$T^{(2)}_{ 0.5, 3}$ & 49.91 (2.99) & 50.08 (4.157) & 50 (0)   &  26.73 (5.186) & 27.87 (7.324) & 25.07 (0.303) \\
$T^{(2)}_{ 0.5, 4}$ & 50 (3.077) & 50.2 (4.115) &  50 (0)  & 26.93 (5.39)  & 27.96 (7.286)  & 25.08 (0.308)  \\

$T^{(2)}_{ 1, 1}$ & 49.48 (2.387) & 49.64 (3.593) & 50 (0)  & 27.25 (5.273) & 29.7 (8.063) & 25.23 (0.611) \\
$T^{(2)}_{ 1, 1.5}$ & 49.66 (2.414) & 49.77 (3.615) & 50 (0) & 27.48 (5.419) & 29.56 (7.684) & 25.27 (0.677) \\
$T^{(2)}_{ 1, 2}$ & 49.73 (2.482) & 49.82 (3.716) & 50 (0)  & 27.69 (5.781) & 29.69 (7.708) & 25.29 (0.696) \\
$T^{(2)}_{ 1, 3}$ & 49.92 (2.603) & 50.07 (3.661) & 50 (0)    & 28.31 (6.111)  & 29.57 (7.832)  & 25.31 (0.728)  \\
$T^{(2)}_{ 1, 4}$ & 49.99 (2.662) & 50.15 (3.615) &  50 (0)  &  28.55 (6.340) & 29.70 (7.885)  & 25.32 (0.749) \\
\bottomrule
\end{tabular}
\end{table}

\begin{table}[htbp]
\caption{$p$-values of novel tests - U.S. GNP data}
\label{pvalUSGNP}\centering
\begin{tabular}{@{}lllllllll@{}}\\
\toprule 
Statistic & $\mathcal{J}_{0.5, 0.5}$ & $\mathcal{J}_{1, 0.5}$ & $\mathcal{J}_{0.5, 1}$ & $\mathcal{J}_{1, 1}$ & $\mathcal{J}_{0.5, 2}$ & $\mathcal{J}_{1, 2}$ & $\mathcal{L}_{0.5, 0.5}$ & $\mathcal{L}_{1, 0.5}$ \\ \midrule
$p$-values & 0.0050 & 0.0040 & 0.0110 & 0.0030 & 0.0090 & 0.0010 & 0.0020 & 0.0020 \\
Position & 1985 (3) & 1985 (3) & 1985 (3) & 1985 (3) & 1984 (2) & 1984 (2) & 1985 (3) & 1985 (3) \\\midrule 
Statistic & $\mathcal{L}_{0.5, 1}$ & $\mathcal{L}_{1, 1}$ & $\mathcal{L}_{0.5, 2}$ & $\mathcal{L}_{1, 2}$ & $T^{(1)}_{0.5, 1}$ & $T^{(1)}_{0.5, 1.5}$ & $T^{(1)}_{0.5, 2}$ & $T^{(1)}_{0.5, 3}$ \\ \midrule
$p$-values & 0.0040 & 0.0050 & 0.0040 & 0.0060 & 0.0080 & 0.0050 & 0.0090 & 0.0120 \\
Position & 1985 (3) & 1985 (3) & 1985 (3) & 1985 (3) & 1984 (2) & 1984 (2) & 1984 (2) & 1984 (2) \\\midrule 
Statistic & $T^{(1)}_{0.5, 4}$ & $T^{(1)}_{1, 1}$ & $T^{(1)}_{1, 1.5}$ & $T^{(1)}_{1, 2}$ & $T^{(1)}_{1, 3}$ & $T^{(1)}_{1, 4}$ & $T^{(2)}_{0.5, 1}$ & $T^{(2)}_{0.5, 1.5}$ \\ \midrule
$p$-values & 0.0140 & 0.0010 & 0.0050 & 0.0020 & 0.0050 & 0.0020 & 0.0100 & 0.0070 \\
Position & 1984 (2) & 1984 (2) & 1984 (2) & 1984 (2) & 1984 (2) & 1984 (2) & 1984 (2) & 1984 (2) \\\midrule 
Statistic & $T^{(2)}_{0.5, 2}$ & $T^{(2)}_{0.5, 3}$ & $T^{(2)}_{0.5, 4}$ & $T^{(2)}_{1, 1}$ & $T^{(2)}_{1, 1.5}$ & $T^{(2)}_{1, 2}$ & $T^{(2)}_{1, 3}$ & $T^{(2)}_{1, 4}$ \\ \midrule
$p$-values & 0.0070 & 0.0100 & 0.0100 & 0.0030 & 0.0020 & 0.0070 & 0.0050 & 0.0040 \\
Position & 1984 (2) & 1984 (2) & 1984 (2) & 1984 (2) & 1984 (2) & 1984 (2) & 1984 (2) & 1984 (2) \\
\bottomrule
\end{tabular}
\end{table}

\begin{table}[htbp]
\caption{$p$-values - Argentina rainfall data}
\label{pvalArgentina} \centering
\begin{tabular}{@{}lllllllll@{}}\\
\toprule 
Statistic & $\mathcal{J}_{0.5, 0.5}$ & $\mathcal{J}_{1, 0.5}$ & $\mathcal{J}_{0.5, 1}$ & $\mathcal{J}_{1, 1}$ & $\mathcal{J}_{0.5, 2}$ & $\mathcal{J}_{1, 2}$ & $\mathcal{L}_{0.5, 0.5}$ & $\mathcal{L}_{1, 0.5}$ \\ \midrule
$p$-values & 0.0080 & 0.0160 & 0.0030 & 0.0090 & 0.0030 & 0.0080 & 0.0030 & 0.0030 \\
Position & 1956 & 1956 & 1956 & 1956 & 1956 & 1956 & 1956 & 1956 \\\midrule 
Statistic & $\mathcal{L}_{0.5, 1}$ & $\mathcal{L}_{1, 1}$ & $\mathcal{L}_{0.5, 2}$ & $\mathcal{L}_{1, 2}$ & $T^{(1)}_{0.5, 1}$ & $T^{(1)}_{0.5, 1.5}$ & $T^{(1)}_{0.5, 2}$ & $T^{(1)}_{0.5, 3}$ \\ \midrule
$p$-values & 0.0020 & 0.0040 & 0.0030 & 0.0020 & 0.0659 & 0.0939 & 0.0959 & 0.0430 \\
Position & 1956 & 1956 & 1956 & 1956 & 1910 & 1910 & 1953 & 1956 \\\midrule 
Statistic & $T^{(1)}_{0.5, 4}$ & $T^{(1)}_{1, 1}$ & $T^{(1)}_{1, 1.5}$ & $T^{(1)}_{1, 2}$ & $T^{(1)}_{1, 3}$ & $T^{(1)}_{1, 4}$ & $T^{(2)}_{0.5, 1}$ & $T^{(2)}_{0.5, 1.5}$ \\ \midrule
$p$-values & 0.0210 & 0.2118 & 0.1179 & 0.0779 & 0.0539 & 0.0330 & 0.0519 & 0.0719 \\
Position & 1956 & 1953 & 1953 & 1953 & 1956 & 1956 & 1910 & 1914 \\\midrule 
Statistic & $T^{(2)}_{0.5, 2}$ & $T^{(2)}_{0.5, 3}$ & $T^{(2)}_{0.5, 4}$ & $T^{(2)}_{1, 1}$ & $T^{(2)}_{1, 1.5}$ & $T^{(2)}_{1, 2}$ & $T^{(2)}_{1, 3}$ & $T^{(2)}_{1, 4}$ \\ \midrule
$p$-values & 0.0939 & 0.0929 & 0.1039 & 0.1548 & 0.1319 & 0.1139 & 0.0719 & 0.0749 \\
Position & 1914 & 1953 & 1953 & 1914 & 1953 & 1953 & 1953 & 1953 \\
\bottomrule
\end{tabular}
\end{table}

\newpage

\begin{sidewaystable}
\caption{Detected change-points in S\&P~500 absolute log returns for different values of $\gamma$.}
\label{tab:sp500cp}
  \begin{tabular}{@{}lllp{6in}}

  \toprule
  \# & Date & Detected by & Event  \\
  \midrule
  1 & 2018-02-01 & All ($\gamma=0.5$ and $1$) & ``Volmageddon'': the S\&P~500 fell $\approx$4\% and the VIX surged $\approx$20 points on 5 February (its largest single-day jump since 1987) as volatility-linked ETPs unwound and rising-rate fears intensified \cite{AMF2018VIX, BIS2018Q1}. \\
  2 & 2018-04-17 & All ($\gamma=0.5$ and $1$) & Syria-strike risk, Russia tensions, bank earnings, and trade-war concerns drove elevated volatility around 12--16 April; easing geopolitical fears positive earnings subsequently stabilised markets \cite{CNBC2018April, Reuters2018Syria}.  \\
  3 & 2018-07-12 & All ($\gamma=0.5$ and $1$) & U.S.\ tariffs on \$34bn of Chinese goods took effect on 6 July and China retaliated immediately; Bloomberg's subsequent report of a further \$200bn tariff list renewed trade-war volatility \cite{Bloomberg2018AsianStocks, CNBC2018Tariffs}.  \\
  4 & 2018-10-09 & All ($\gamma=0.5$ and $1$) & On 10 October the S\&P~500 fell 3.3\%, the Dow dropped $\approx$830 points, and the Nasdaq~100 posted its worst session since 2011; October
  became the S\&P~500's worst month since 2011 with nearly \$2tn in market value lost \cite{Bloomberg2018October, CNBC2018Stock}.  \\
  5 & 2019-01-22 & \shortstack[l]{All ($\gamma=0.5$ and $1$) \\ exc. $\mathcal{L}_{\cdot,1}$, $\mathcal{L}_{\cdot,1.5}$}& Equities rebounded sharply from the late-2018 sell-off as the Fed signaled patience on rate hikes and U.S.--China trade optimism improved, producing the strongest January rally in years \cite{Reuters2019FedPivot, Schroders2019January}.  \\
  6 & 2019-01-31 & $\mathcal{L}_{\cdot,1}$, $\mathcal{L}_{\cdot,1.5}$ only & Same regime shift detected six trading days later: the S\&P~500 posted its best January since 1987 after strong earnings and the Fed's explicit ``patient'' guidance on rate hikes \cite{CNBC2019January}.  \\
  7 & 2020-02-20 & All ($\gamma=0.5$ and $1$)& A sudden risk-off sell-off on 20 February - one day after the S\&P~500 and Nasdaq hit all-time highs - marked the onset of the
  COVID-19 market crash, which deepened into a month-long pandemic-driven collapse \cite{CNBC2020February, Reuters2020YearHistory}.  \\
  8 & 2020-04-06 & All ($\gamma = 1$) & U.S. equities staged a sharp relief rally as investors reacted to early signs that COVID-19 deaths/infections were slowing in some hard-hit regions, including New York, Italy and Spain. The S\&P 500 rose about 7\%, its biggest jump in nearly two weeks \cite{Reuters2020WallStreet}.\\
  9 & 2020-04-22 & All ($\gamma=0.5$) & WTI crude futures turned negative for the first time on 20 April as pandemic lockdowns crushed demand and storage capacity filled; the S\&P~500 fell 1.8\% and energy stocks dropped sharply \cite{Bloomberg2020Oil}.  \\
  10 & 2020-06-30 & All ($\gamma = 1$) & End of Q2 2020 rebound; S\&P~500 posts best quarter since 1998 amid stimulus support, reopening optimism, renewed COVID case concerns, and U.S.--China tensions \cite{Reuters2020BestQuarter}.\\
  11 & 2020-11-24  & All ($\gamma=0.5$)& U.S.\ stocks rallied strongly on 24 November, with the Dow crossing 30{,}000 for the first time, as vaccine progress and formal clearance for President-elect Biden's transition drove a broad reopening rotation \cite{Reuters2020Dow30k, Reuters2020Goldman}.  \\
  12 & 2021-03-11 & All ($\gamma = 1$) & Biden signed the \$1.9 trillion American Rescue Plan; S\&P 500 and Dow closed at record highs; tech rebounded while recovery/value rotation remained active \cite{CNBC2021Stimulus, Reuters2021RecordHighs}.\\
  13 & 2021-05-20 & All ($\gamma=0.5$) & Markets fell for three consecutive sessions after Fed minutes revealed policymakers discussing tapering crisis-era asset purchases, compounded by inflation fears and a sharp cryptocurrency sell-off \cite{Reuters2021Taper}.  \\
  14 & 2021-11-24 & All ($\gamma=0.5$ and $1$) & The Omicron variant and Powell's hawkish pivot drove late-November
  volatility; the S\&P~500 recorded its first losing November since 2011 as investors repriced aggressive Fed tightening \cite{ Bloomberg2021Omicron, Reuters2021Omicron}.  \\
  15 & 2022-01-14 &  All ($\gamma=0.5$) & CPI at 7\%, accelerated Fed-hike expectations, and a tech/growth sell-off pushed the S\&P~500 into correction territory and drove the VIX to a 52-week high during a ``wild ride'' month of trading \cite{Nasdaq2022January, Reuters2022WildRide}.  \\
  16 & 2022-11-10 & All ($\gamma=0.5$ and $T^{(1)}_{1, 2}$) & A softer-than-expected October CPI print triggered the S\&P~500's largest single-day gain since April 2020 ($+$5.54\%) and the Nasdaq's largest since 2020 ($+$7.35\%), as investors bet on a less aggressive Fed path \cite{CNBC2022CPI, Reuters2022CPI}.  \\
  17 & 2023-03-22 & All ($\gamma=0.5$ and $1$) & The Fed raised rates 25\,bps amid global banking turmoil following the SVB and Signature Bank failures while signalling an imminent pause; SVB's collapse simultaneously drove the MOVE index to levels last seen during the 2008 financial crisis \cite{Reuters2023SVB, Reuters2023FedHike}.  \\
  \bottomrule
  \end{tabular}%
\end{sidewaystable}

\newpage

\begin{algorithm}
\caption{Binary segmentation with permutation bootstrap}
\label{BSAl}
\begin{algorithmic}[1]
\Require \(x_{1:n}\), \(B\), \(\mathcal{S}\), \(\gamma\), \(a\), 
         \(\alpha\), \(w\), \(o\), \(d\), \(d_{\max}\)
\Ensure Set of detected change-points

\If{\(n < 2w\) \textbf{or} \(d > d_{\max}\)}
    \State \Return \(\emptyset\)
\EndIf

\State
\[
T_{\mathrm{obs}}
=
\max_{w \le k \le n-w}
\mathcal{S}_{\gamma,a}(x_{1:n},k),
\qquad
k^*
=
\arg\max_{w \le k \le n-w}
\mathcal{S}_{\gamma,a}(x_{1:n},k)
\]

\For{\(b=1,\ldots,B\)}
    \State Generate a random permutation \(\pi_b\) of \(\{1,\ldots,n\}\)
    \State Set \(x^{*(b)}_{1:n} \leftarrow (x_{\pi_b(1)},\ldots,x_{\pi_b(n)})\)
    \State
    \[
    T_b^*
    =
    \max_{w \le k \le n-w}
    \mathcal{S}_{\gamma,a}(x^{*(b)}_{1:n},k)
    \]
\EndFor

\State
\[
\hat p
=
\frac{1}
{B+1}\Big(\sum_{b=1}^{B}
\mathbf{1}\{T_b^* \ge T_{\mathrm{obs}}\}+1\Big)
\]

\If{\(\hat p \ge \alpha\)}
    \State \Return \(\emptyset\)
\EndIf

\State \(c \leftarrow (o+k^*,\, o+1,\, o+n,\, T_{\mathrm{obs}},\, \hat p,\, d)\)

\State \(\mathcal{L} \leftarrow
BS(x_{1:k^*}, B, \mathcal{S}, \gamma, a, \alpha, w, o, d+1, d_{\max})\)

\State \(\mathcal{R} \leftarrow
BS(x_{k^*+1:n}, B, \mathcal{S}, \gamma, a, \alpha, w, o+k^*, d+1, d_{\max})\)

\State \Return \(\mathrm{sort}\{\mathrm{unique}(\{c\} \cup \mathcal{L} \cup \mathcal{R})\}\)

\end{algorithmic}
\end{algorithm}

\clearpage

\section{Appendix B}\label{appendixB}

\begin{proof}[Sketch of the proof of Theorem 1] 
 The proof follows the same methodology as that outlined in \cite{huvskova2006change} and \cite{nas2}. 
Here, we provide a brief sketch of the proof for the Hankel-based statistic $\mathcal{J}_{n,\gamma,a}$; the proof for $\mathcal{L}_{n,\gamma,a}$ follows the same lines and is even simpler due to the form of the Laplace transform, which does not involve special functions. {For the remainder of this proof, it is assumed that integration is performed over the maximal domain on which the functions are defined; integration limits are omitted henceforth.}


To obtain the limiting null distribution of the test statistic, we first note that $\mathcal{J}_{n,\gamma,w}$ can be represented as  
$\mathcal{J}_{n, \gamma,w}=\max\limits_{1\leq k< n} c_{n, k}(\gamma)J_{k, n-k,a}$, where
{$c_{n, k}(\gamma)=\Big(\frac{k(n-k)}{n^2}\Big)^\gamma \frac{k(n-k)}{n}$, and

}


{\begin{align*}
   J_{k, n-k,a}&= \frac{1}{k^2a}\sum\limits_{l=1}^k\sum\limits_{m=1}^k I_0\left(\frac{2\sqrt{X_lX_m}}{a}\right)\exp\left(\frac{-(X_l+X_m)}{a}\right)+\\&\frac{1}{(n-k)^2a}\sum\limits_{l=k+1}^n\sum\limits_{m=k+1}^n I_0\left(\frac{2\sqrt{X_lX_m}}{a}\right)\exp\left(\frac{-(X_l+X_m)}{a} \right)-\\&\frac{2}{k(n-k)a}\sum\limits_{l=1}^k\sum\limits_{m=k+1}^n I_0\left(\frac{2\sqrt{X_lX_m}}{a}\right)\exp\left(\frac{-(X_l+X_m)}{a}\right).
\end{align*}
}

The statistic  $\mathcal{L}_{n,\gamma,a}$ can be represented in an analogous way.
{
Both transforms considered here can be represented in the form $E(A(tX))$. For the Laplace transform, $A(tx)=e^{-tx}$, while for the Hankel transform, $A(tx)=J_0(2\sqrt{tx})$. 
Therefore, we provide general asymptotic results for an arbitrary integral transform of this form.

Let

{\begin{align*}
   J_{k, n-k,w}&= \frac{1}{k^2}\sum\limits_{l=1}^k\sum\limits_{m=1}^k\int A(tX_l)A(tX_m)w(t)dt+\\&\frac{1}{(n-k)^2}\sum\limits_{l=k+1}^n\sum\limits_{m=k+1}^n\int A(tX_l)A(tX_m)w(t)dt-\\&\frac{2}{k(n-k)}\sum\limits_{l=1}^k\sum\limits_{m=k+1}^n \int A(tX_l)A(tX_m)w(t)dt.
\end{align*}
}

}

{
Let us also consider
$q_w(x,y)=\int A(tx)A(ty)w(t)dt$, for which assumption $E{q}_w^2(X_1, X_2)<\infty$ holds.} A sufficient condition for this is $\mathbb{E}\!\left[\|A(\cdot X_1)\|_{w}^2\right]<\infty$.

Notice that {these conditions are satisfied} for both the Laplace and Hankel transforms.

{
Let $\Tilde{q}_w(x, y)=q_w(x, y)-E(q_w(X_1, X_2)|X_1=x)-E(q_w(X_1, X_2)|{X_2=y})+Eq_w(X_1, X_2). Then $

we have

\begin{equation}\label{nejednakosti}
E(\Tilde{q}_w(X_1, X_2)|X_1)=E\Tilde{q}_w(X_1, X_2)=0.
\end{equation}

}

{

It should be noted that the functions $q_w(x, y)$ and $\Tilde{q}_w(x, y)$ are symmetric in their arguments.

Since $E\Tilde{q}_w^2(X_1, X_2)<\infty$, there exist orthogonal eigenfunctions ${f^w_s(t),\; s=1, 2, \dots}$ and corresponding eigenvalues {${\lambda^w_s,\; s=1, 2, \dots}$} such that the following spectral approximation holds; see \cite{borovskikh1994theory, serfling2009approximation}: }

{

\begin{align}
   \lim  \limits_{K\to\infty} \int \int (\Tilde{q}_w(x_1, x_2)-\sum\limits_{s=1}^K\lambda^w_s f^w_s(x_1)f^w_s(x_2))^2 dF(x_1)dF(x_2)=0,\label{eigenvalues}
   \end{align}

}

   where
 {
  \begin{align*}
  & \int {(f^{w}_s(x_1))^2}dF(x_1)=1, \; s=1, 2, \dots;\\
  & \int f_r^w(x_1)f_s^w(x_1)dF(x_1)=0, \;r\neq s=1, 2, \dots.
 \end{align*} 
 }
 
 Moreover,

{
\begin{align}\label{eigen}
  &E\Tilde{q}_w^2(X_1,X_2) =\sum\limits_{j=1}^\infty (\lambda_j^w)^2.
\end{align}
}
{ Furthermore, \(q_w\) is a positive semidefinite kernel, since
\[
q_w(x,y)
=
\langle A(\cdot x),A(\cdot y)\rangle_w.
\]
Its centered version \(\widetilde q_w\) is also positive semidefinite.
Consequently,
$\lambda_s^w\geq 0,\qquad s=1,2,\ldots.$

}
{
In addition, if 
$E q_w^2(X_1,X_1)<\infty$ holds, ( for which a sufficient condition is $E\|A(\cdot X_1)\|_w^4<\infty)$, we can conclude that }
{ $\widetilde q_w$ is positive semidefinite and
$E\widetilde q_w(X_1,X_1)<\infty$, the corresponding integral
operator is of a trace class. Therefore, in this case,
\begin{align}\label{trace}
\sum_{s=1}^{\infty}\lambda_s^w
=
E\widetilde q_w(X_1,X_1)
=
E q_w(X_1,X_1)-E q_w(X_1,X_2),
\end{align}
which allows the further simplification of \eqref{asimptotika1} and \eqref{asimptotika2}.
}

We now divide the proof into several steps. 

In Step (I), we decompose the test statistic into four parts and show that two of them do not affect the limiting distribution. 
This is established by repeatedly applying the H\'ajek-R\'enyi inequality to control the relevant terms.

In Step (II), we determine the finite-dimensional distributions of the limiting process by establishing an approximation formula and applying Donsker's theorem.

In Step (III), using the properties of the Wiener process, we determine the desired asymptotic distribution and complete the proof. 
\medskip

\noindent\underline{Step (I):} The following decomposition holds:
{
\begin{equation*}
    J_{k, n-k,w}=C_{k1}+C_{k2}+C_{k3}+C_{k4},
\end{equation*}

}
where 

{

\begin{align*}
    C_{k1}&=\frac{n}{k(n-k)}\Big(\frac{1}{k}\sum\limits_{v=1}^k\sum\limits_{\substack{s=1\\ s\neq v}}^k \Tilde{q}_w(X_v, X_s)+\frac{1}{n-k}\sum\limits_{v=k+1}^n\sum\limits_{\substack{s=k+1\\ s\neq v}}^n \Tilde{q}_w(X_v, X_s)-\frac{1}{n}\sum\limits_{v=1}^n\sum\limits_{\substack{s=1\\ s\neq v}}^n \Tilde{q}_w(X_v, X_s)\Big), \\
    C_{k2}&=\frac{n}{k(n-k)}\Big({ E}\int A^2(tX_1)w(t)dt-Eq_w(X_1, X_2)\Big),\\
    C_{k3}&=-\frac{2}{k^2}\sum\limits_{r=1}^k(E(q_w(X_r, X)|X_r)-Eq_w(X_1, X_2))\\&-\frac{2}{(n-k)^2}\sum\limits_{r=k+1}^n (E(q_w(X_r, X)|X_r)-Eq_w(X_1, X_2)),\\
    C_{k4}&=\frac{1}{k^2}\sum\limits_{r=1}^k \Big(\int A^2(tX_r)w(t)dt)-E\int A^2(tX_r)w(t)dt\Big)\\&+\frac{1}{(n-k)^2}\sum\limits_{r=k+1}^n \Big(\int A^2(tX_r)w(t)dt)\!-\!E\int A^2(tX_r)w(t)dt\Big),
\end{align*}
and $X$ is a random variable independent of $X_r$ with distribution $F_1$.

}

We begin by showing that
$C_{k3}$ and $C_{k4}$ do not affect the limiting distribution. 

The random variables

{

\begin{equation*}
    L_{r;w}=E(q_w(X_r, X)|X_r)-Eq_w(X_1, X_2), r=1, 2, \dots, n
\end{equation*}

}
are independent and identically distributed. 
Moreover, under the condition $\mathbb E\,\widetilde q_w^{\,2}(X_1,X_2)<\infty$, they have finite variance.


By applying the H\'ajek-R\'enyi inequality, we obtain that, for every positive constant $A$ and every $\gamma \in [0, 1]$, the following inequality holds:

{

\begin{equation*}
    P_{H_0}\Big(\max\limits_{1\leq k < n} c_{k, n}(\gamma) \frac{2}{k^2}\Big|\sum\limits_{r=1}^k L_{r;w}\Big | \geq A\Big)\leq D_3n^{-\min({2\gamma}, 1)}(1+Ind(\gamma\geq\frac{1}{2})\log(n)).
\end{equation*}
}


{

\begin{equation*}
    P_{H_0}\Big(\max\limits_{1\leq k < n} c_{k, n}(\gamma) \frac{2}{(n-k)^2}\Big|\sum\limits_{r=k+1}^n L_{r;w}\Big | \geq A\Big)\leq D_3n^{-\min({2\gamma}, 1)}(1+Ind(\gamma\geq \frac{1}{2})\log(n)).
\end{equation*}
}

Consequently,
\begin{equation*}
    \max\limits_{1\leq k<n} c_{k, n}(\gamma)|C_{k3}|=o_P(1), n\to\infty, 
\end{equation*}
and the term $|C_{k3}|$ does not affect the limiting distribution.
The fact that $C_{k4}$ does not affect the limiting behavior can be shown analogously.

Furthermore, note that

{
\begin{align*}
Var(C_{k1})=O\Big(2\Big(\frac{n}{k(n-k)}\Big)^2E\Tilde{q}^2_w(X_1, X_2)\Big).
\end{align*}
}

Since $\sqrt{Var(C_{k1})}$ and $C_{k2}$ are of the same order, it remains to determine the limiting distribution of $C_{k1}$ in Step (II). 

\underline{Step (II):} 
Consider the auxiliary statistic

{
\begin{equation*}
    S_k(\Tilde{q}_w)=\sum\limits_{1\leq i < j \leq k} \Tilde{q}_w(X_i, X_j), \; k = 1, 2, \dots, n.
\end{equation*}

}


Applying the H\'ajek-R\'enyi inequality yields the following bound:

{

\begin{align*} &P\Big(\max\limits_{1\leq k < n} n\Big(\frac{k(n-k)}{n^2}\Big)^{\gamma+1} \frac{1}{k^2}|S_k(\Tilde{q}_w)|\geq A\Big)
   \leq \frac{D_0}{A^2}E\Tilde{q}^2_w(X_1, X_2),
\end{align*}
}

for some positive constant $D_0$. Repeating the computation for the function 
{
$\Tilde{S}_k(\Tilde{q}_w)=\sum\limits_{k+1\leq i<j\leq  n} \Tilde{q}_w(X_i, X_j)$,
}
and applying the H\'ajek-R\'enyi inequality twice gives

{
\begin{equation*}
    P\Big(\max\limits_{1 \leq k < n} c_{k, n}(\gamma)|C_{k1}|\geq A\Big)\leq \frac{D_1}{A^2} E\Tilde{q}^2_w(X_1, X_2), 
\end{equation*}
}
for some positive constant $D_1$. This inequality remains valid if $\Tilde{q}_w$ is replaced by any function satisfying \eqref{nejednakosti}. Therefore, we apply it to the function
{ $\Tilde{q}_w-\Tilde{q}_{wK}$}, where $\Tilde{q}_{wK}$ is defined by
\begin{equation*}
    \Tilde{q}_{wK}(x, y)=\sum\limits_{s=1}^K \lambda^w_s f^w_s(x)f^w_s(y),
\end{equation*}

where { $(\lambda^w_s, f^w_s)$} are the sorted eigenvalues and eigenfunctions defined in \eqref{eigenvalues}, and $K$ is an arbitrary natural number. 

{Moreover, by using \eqref{eigen} and applying the H\'ajek-R\'enyi inequality once more, we obtain
\begin{equation}\label{nejednakostC}
    P\Big(\max\limits_{1\leq k < n} c_{k, n}(\gamma) |C_{k1}- C_{k1}(K)|\geq A\Big)\leq \frac{D_5}{A^2}E(\Tilde{q}_w(X_1, X_2)-\Tilde{q}_{wK}(X_1, X_2))^2=
    \frac{D_5}{A^2}\sum\limits_{j=K+1}^\infty (\lambda_j^w)^{2},
\end{equation}
where $D_5$ is a positive constant, $n\geq 2$, $K\in \mathbb{N}$, and $C_{k1}(K)$ denotes $C_{k1}$ with $\Tilde{q}$ replaced by $\Tilde{q}_{wK}$. Therefore,
\begin{align*}
    \frac{k(n-k)}{n}C_{k1}(K)&=\sum\limits_{s=1}^K\lambda^w_s\Big(\frac{n}{k(n-k)}\Big(\sum\limits_{j=1}^k f^w_s(X_j)-\frac{k}{n}\sum\limits_{j=1}^n f^w_s(X_j)\Big)^2\\& - \frac{1}{n}\Big(\frac{n-k}{k}\sum\limits_{j=1}^k (f^{w}_s(X_j))^2+\frac{k}{n-k}\sum\limits_{j=k+1}^n (f^w_s(X_j))^2\Big)\Big). 
\end{align*}}
Using Donsker's theorem \cite{billingsley2013convergence}, we obtain
\begin{equation*}
    \Big\{\frac{1}{\sqrt{n}}\sum\limits_{i=1}^{[nt]} \mathbf{f_K}^w(X_i), t\in(0,1)\Big\}\xrightarrow[\;n\to\infty\;]{D((0,1),\mathbb{R}^K)} {\mathbf{W_K}^w(t), t\in(0,1)},
\end{equation*}
where $\mathbf{W_K}^w=(W_1^w, \dots, W_K^w)$ is a $K$-dimensional Wiener process with independent components in $\mathbb{R}^K$, and $\mathbf{f_K}^w=(f^w_1, \dots, f^w_K)$ is the vector of eigenfunctions up to index $K$. Therefore, the asymptotic distribution of $\max\limits_{1\leq k < n} \Big(\frac{k(n-k)}{n^2}\Big)^{\gamma}C_{k1}(K)$ is given by
\begin{equation*}
\max\limits_{1\leq k < n} \Big(\frac{k(n-k)}{n^2}\Big)^{\gamma}{\frac{k(n-k)}{n}}\sum\limits_{s=1}^K\lambda_s^w\Big(\frac{n^2}{k(n-k) }(W_s^w(k/n)-\frac{k}{n}W_s^w(1))^2-1\Big).
\end{equation*}

\underline{Step (III):} Note that the zero-mean random processes $\lambda_s^w\Big(\frac{n^2}{k(n-k) }(W^w_s(k/n)-\frac{k}{n}W^w_s(1))^2-1\Big)$ are independent for $s=K+1, K+2, \dots$.  
Furthermore, since their variances are finite, we can apply the H\'ajek-R\'enyi inequality to obtain
\begin{align*}
     P\Big(\max\limits_{1\leq k < n} \Big(\frac{k(n-k)}{n^2}\Big)^{\gamma}\Big |\!\!\sum\limits_{s=K+1}^\infty \!\!\lambda_s^w\Big(\frac{n^2}{k(n-k) }(W_s^w(k/n)-\frac{k}{n}W_s^w(1))^2-1\Big) \Big |\geq A\Big)
   \leq \frac{D_5}{A^2}\!\!\sum\limits_{s=K+1}^\infty\!\! (\lambda_s^w)^2.
\end{align*}

From \eqref{nejednakostC}, we see that, for sufficiently large $K$, the asymptotic distributions of appropriately scaled $C_{k1}$ and $C_{k1}(K)$ are arbitrarily close. Consequently, by letting $K$ tend to infinity, we conclude that the asymptotic distribution of $\max\limits_{1\leq k < n}\Big(\frac{k(n-k)}{n^2}\Big)^{\gamma} C_{k1}$ coincides with that of
\begin{equation*}
   \max\limits_{t\in (0, 1)} \Big(t(1-t)\Big)^{\gamma}\sum\limits_{s=1}^\infty \lambda_s^w\Big(\frac{(W^w_s(t)-tW^w_s(1))^2}{t(1-t)}-1\Big).
\end{equation*}
This completes the proof.
\end{proof}
\end{document}